\documentclass[leqno]{article}
\usepackage{times}
\usepackage[pdftex]{graphicx}
\usepackage{amssymb,amsfonts,amsmath,amsthm}
\usepackage{float}
\usepackage{xcolor}

\usepackage{natbib}
\setcitestyle{authoryear,open={(},close={)}}

\numberwithin{equation}{section}

\newtheorem{theorem}{Theorem}[section]
\newtheorem{proposition}{Proposition}[section]
\newtheorem{lemma}{Lemma}[section]
\newtheorem{ass}{Assumption}[section]
\newtheorem{definition}{Definition}[section]
\newtheorem{remark}{Remark}[section]
\newtheorem{cor}{Corollary}[section]

\newtheorem{notation}{Notational Convention}[section]

\newcommand\EE {\mathbb E}
\newcommand\FF {\mathbb F}

\newcommand\RR {\mathbb R}
\newcommand\PP {\mathbb P}

\def\bone{\mathbf{1}}

\def\qed{\hskip6pt\vrule height6pt width5pt depth1pt}

\def\qed{\hskip 6pt\vrule height6pt width5pt depth1pt}

\newcommand{\ed}{\end{document}}
\newcommand{\be}{\begin{equation}}
\newcommand{\ee}{\end{equation}}
\newcommand{\bq}{\begin{eqnarray}}
\newcommand{\eq}{\end{eqnarray}}

\vspace{4in}

\newcommand{\een}{\EE^\alpha_{N-1}}

\begin{document}

\title{Liquidity Effects of Trading Frequency \thanks{This version: February 13, 2017. First version August 28, 2015.}}

\author{Roman Gayduk and Sergey Nadtochiy\footnote{Partial support from the NSF grant DMS-1411824 is acknowledged by both authors.} \footnote{We thank the anonymous referees and the Associate Editor for constructive comments that helped us improve the paper significantly.} \footnote{Address the correspondence to Sergey Nadtochiy, Mathematics Department, University of Michigan, 530 Church Street, Ann Arbor, MI 48109, USA; e-mail: sergeyn@umich.edu.}\\$\,\,\,\,$\\ \emph{University of Michigan}}
\date{}
\maketitle

%\title{Liquidity Effects of Trading Frequency.\footnote{Partial support from the NSF grant DMS-1411824 is acknowledged by both authors.}}
%\author{Roman~Gayduk and Sergey~Nadtochiy\footnote{Address the correspondence to: Mathematics Department, University of Michigan, 530 Church St, Ann Arbor, MI 48104; sergeyn@umich.edu.}\footnote{We thank the anonymous referees and the Associate Editor for the constructive comments that helped us improve the paper significantly.}}
%\date{Current version: Feb 13, 2017\\
%Original version: Aug 28, 2015
%}
%\maketitle

\begin{abstract}
In this article, we present a discrete time modeling framework, in which the shape and dynamics of a Limit Order Book (LOB) arise endogenously from an equilibrium between multiple market participants (agents).
%I will start by reviewing the various models LOB and will mention some real-world phenomena that illustrate the need for a different approach. 
%The new framework captures very closely the true, micro-level, mechanics of an auction-style exchange. At the same time, it uses the standard abstractions of a continuum-player game to obtain a tractable macro-level description of the LOB. 
We use the proposed modeling framework to analyze the effects of trading frequency on market liquidity in a very general setting. In particular, we demonstrate the dual effect of high trading frequency. On the one hand, the higher frequency increases market efficiency, if the agents choose to provide liquidity in equilibrium. On the other hand, it also makes markets more fragile, in the sense that the agents choose to provide liquidity in equilibrium only if they are market-neutral (i.e., their beliefs satisfy certain martingale property). Even a very small deviation from market-neutrality may cause the agents to stop providing liquidity, if the trading frequency is sufficiently high, which represents an endogenous liquidity crisis (aka flash crash) in the market. This framework enables us to provide more insight into how such a liquidity crisis unfolds, connecting it to the so-called adverse selection effect.
\end{abstract}

\begin{small}{\bf Key words}: liquidity, trading frequency, Limit Order Book, continuum-player games, conditional tails of It\^o processes.\end{small}

\section{Introduction}

%The technological development presents new challenges to the mathematical modeling of social and economic phenomena. In particular, the rapid growth of electronic trading has changed significantly the existing approaches to modeling financial markets. The classical mathematical models used to focus on the macroscopic description of the financial processes, often, taking as input the price levels at which certain assets can be purchased or sold.  However, in reality, the price arises as an outcome of the interaction between market participants, and understanding the mechanics of the price formation process has become an important problem on its own.
This paper is concerned with liquidity effects of trading frequency on an auction-style exchange, in which the participating agents can post limit or market orders. 
%trading, which accounts for more than a half of all US equity trading volume. 
%These traders buy and sell the asset in relatively small quantities but at a very high rate, making significant profits due to the large overall volume. 
On the one hand, higher trading frequency provides more opportunities for the market participants to trade, hence, improving the liquidity of the market and increasing its efficiency. 
On the other hand, higher trading frequency also provides more opportunities for some participants to manipulate the price and disrupt the market liquidity.
Such a manipulation creates a new type of risk, which reveals itself in unusually high price deviations, which cannot be explained by changes in the fundamental value of the asset. The most famous example of this phenomenon is the ``flash crash'' of $2010$. 
This example motivates the need for a comprehensive study of
%the market microstructure that would focus on describing the behavior of a large number of agents participating in the price formation. 
%In particular, such a study needs to investigate
the tradeoff between the liquidity providing role of strategic players and the liquidity risk they generate, and its relation to trading frequency.
%In this paper, we analyze market microstructure in the context of an auction-style exchange (as most exchanges currently are), in which the participating agents can post limit or market orders.
The collective liquidity of the market is captured by the \emph{Limit Order Book (LOB)}, which contains all the limit buy and sell orders. 

The goal of the present paper is two-fold. 
First, we develop a new framework for modeling market microstructure,
%which strikes the right balance between generality and tractability. In particular, we propose a modeling framework
in which the shape of the LOB, and its dynamics, arise \emph{endogenously} from the interactions between the agents.
%This is in contrast to many of the existing results on market microstructure, which assume that the shape and dynamics of the LOB are given exogenously. 
Among the many advantages of such an approach is the possibility of modeling the market reaction to changes in the rules of the exchange: e.g., limited trading frequency, transaction tax, etc. 
The second, and most important, goal of the present work is to investigate the \emph{liquidity effects} of \emph{trading frequency}, using the proposed modeling framework.  
%In addition to the introduction of a novel modeling framework, the main contribution of this paper is the description of liquidity effects of the market participants as the trading frequency grows. 
In particular, the main results of this paper (cf. the discussion in Section \ref{se:examples}, as well as Theorems \ref{le:main.zeroTermSpread}, \ref{thm:main.necessary} and Corollary \ref{prop:main.smallspread}, in Section \ref{se:main}) describe the dual effect of high trading frequency. On the one hand, if the agents choose to provide liquidity in equilibrium, higher trading frequency decreases the bid-ask spread and makes the expected profits of all market participants converge to the same (fundamental) value, thus, improving the market efficiency. 
On the other hand, higher trading frequency also makes the LOB more sensitive to the deviations of the agents' attitudes from market-neutrality. 
It is, of course, clear that a strong bullish or bearish signal induces market participants to trade at a higher or lower price. However, the novelty of our observation is in the role that the trading frequency plays in amplifying this effect. Namely, we show that, if the trading frequency is high, even if agents have plenty of inventory, a very small deviation from market-neutrality may cause them to stop providing liquidity, by either withdrawing from the market completely, or by posting limit orders far away from the fundamental price. Such actions cause disproportional deviations in the LOB, which cannot be explained by any fundamental reasons: they are much higher than the trading signal (i.e., the expected change in the fundamental price), and they occur without any shortage of supply or demand for the asset. We refer to such a deviation as an \emph{endogenous} liquidity crisis, because it is due to the trading mechanism (i.e., the rules by which the market participants interact), rather than any fundamental reasons (note the similarity with the flash crash). Our framework provides insights into how such a liquidity crisis unfolds, connecting it to the so-called \emph{adverse selection} effect. In particular, Section \ref{se:examples} constructs an equilibrium in which an endogenous liquidity crisis does not occur because of an abnormally large market order, wiping out the liquidity on one side of the LOB, but because the optimal strategies of the agents require them to stop providing liquidity on one side of the LOB. On the mathematical side, our analysis uses the properties of conditional tails of the increments of a general It{\^o} process. The main result in that regard, in Lemma \ref{le:necessary.marginal.maximum}, provides a uniform exponential bound on the conditional tails of the increments of a general It{\^o} process. We believe that this result is useful in its own right, and, to the best of our knowledge, it is not available in the existing literature.

In recent years, we observed an explosion in the amount of literature devoted to the study of market microstructure. In addition to various empirical studies, a large part of the existing theoretical work focuses on the problem of optimal execution:
%, in which an investor needs to liquidate her position in the asset within a given time horizon, by submitting smaller (limit or market) orders and aiming to maximize the profits. 
%The relevant publications, include,
see, among others, \cite{MMS1}, \cite{MMS2}, \cite{MMS6}, \cite{MMS7}, \cite{MMS11}, \cite{MMS13}, \cite{MMS15}, \cite{MMS16}, \cite{MMS20}, \cite{MMS22}, \cite{MMS23}, \cite{MMS26}, \cite{MMS25}, and references therein. 
In these articles, the dynamics and shape of the LOB are modeled exogenously, or, equivalently, the arrival processes of the limit and market orders are specified exogenously. In particular, none of these articles attempts to explain the shape and dynamics of the LOB, arising directly from the interaction between the market participants. 
A different approach to the analysis of market microstructure has its roots in the economic literature. 
For example, \cite{MMS.g1}, \cite{MMS.g2}, \cite{MMS.g3}, \cite{MMS.g4}, \cite{MMS.g5}, \cite{MMS.g6}, \cite{DA.DuZhu}, \cite{Bressan1}, \cite{Bressan2}, \cite{Bressan4} consider equilibrium models of market microstructure, and they are more closely related to the present work. However, the models proposed in the aforementioned papers do not aim to represent the mechanics of an auction-style exchange with sufficient precision, and, in particular, they are not well suited for analyzing the liquidity effects of trading frequency, which is the main focus of the present paper.
%A thorough analysis of an equilibrium-based model for LOB in a single period (i.e. in a static model) is provided in \cite{Bressan1}, \cite{Bressan2}, \cite{Bressan3}, \cite{Bressan4}, but without addressing the specific question of liquidity effects.
%The liquidity role of the agents is analyzed, e.g., in \cite{MMS.gliq1}, \cite{MMS.gliq2}, \cite{MMS.gliq3}, \cite{MMS.gliq4},  but not in the context of market microstructure. The results of the latter papers demonstrate that, depending on the model parameters, the agents may either serve as liquidity providers for each other, or attempt to manipulate the price, reducing the overall liquidity. 
%(compare to the ``flash crash'' described above).
A somewhat related strand of literature focuses on the endogenous formation of LOB in markets with a designated market maker: see e.g., \cite{MMS.gmm1}, \cite{MMS.gmm2}, \cite{MMS.gmm3}, \cite{MMS.gmm4}, \cite{MMS.gmm5}. In these papers, the LOB is not an outcome of a multi-agent equilibrium: instead, it is controlled by a single agent, the market maker. 
%which is not the case in many modern exchanges and, in particular, is not assumed in the present paper (as we study the \emph{auction-style} exchanges).
%Finally, several recent papers have applied an equilibrium-based approach to the problem of optimal execution (cf. \cite{MMS.goe1}, \cite{MMS.goe2}). These papers describe an equilibrium between several agents solving an optimal execution problem, with the LOB (or, the market), against which these agents trade, being specified exogenously, rather than being modeled as an output of the equilibrium.
In the present paper, we model the entire LOB as an output of an equilibrium between a large number of agents, each of whom is allowed to both consume and provide liquidity (in particular, we have no designated market maker).
Our setting is related to the literature on \emph{double auctions} (cf. \cite{DA.Vayanos}, \cite{DA.DuZhu}), with the crucial difference that the participants of each auction are allowed to choose two ``asymmetric" types of strategies: market or limit orders. In addition, the present framework assumes that, ex ante, all agents have access to the same information, and, in this sense, it is similar to \cite{MMS.g1}, \cite{MMS.g3}, \cite{MMS.g6}. In particular, the \emph{adverse selection} effect, herein, does not arise from any a priori information asymmetry between agents, instead, it is caused by the \emph{mechanics} of the exchange.
We formulate the problem as a \emph{continuum-player game} -- this abstraction allows us to obtain computationally tractable results (cf. \cite{Aumann}, \cite{Schmeidler}, \cite{GCarmona} for the concept of a continuum-player game, and \cite{MFG1}, \cite{MFG2}, \cite{MFG3}, \cite{MFG4} for the subclass of mean field games). 
%The connection between our approach and the finite-player games and mean field games is discussed in Subsection \ref{subse:finPlayer} (cf. \cite{MFG1}, \cite{MFG2}, \cite{MFG3}, \cite{MFG4}, for more on mean field games).
%As the state processes of individual agents interact only through the empirical distribution of their controls, our approach falls within the framework of mean field games. The latter topic has received a lot of attention in the recent years: we refer the interested reader to \cite{MFG1}, \cite{MFG2}, \cite{MFG3}, \cite{MFG4}, and references therein, for more information on the subject. It is worth mentioning, however, that, to the best of our knowledge, the type of mean field games presented herein are different from any of the models that have been considered in the literature. As a result, we cannot use much of the existing machinery of mean field games, and have to prove all the desired results by hand.

The paper is organized as follows. Subsection \ref{se:setup} describes the probabilistic setting, along with the execution rules of the exchange and the resulting state processes of the agents. 
Subsection \ref{se:equil.def} defines the equilibrium and introduces the notion of \emph{degeneracy} of the market (which represents an endogenous liquidity crisis).
%Subsection \ref{subse:finPlayer} discusses the connection of our approach to other modeling frameworks.
%Section \ref{se:representation} contains technical results which are needed for a tractable description of the equilibria.
In Section \ref{se:examples}, we construct an equilibrium in a simple model, illustrating how an endogenous liquidity crisis unfolds, and how it is connected to the adverse selection effect.
Theorems \ref{le:main.zeroTermSpread}, \ref{thm:main.necessary}, and Corollary \ref{prop:main.smallspread}, in Section \ref{se:main}, are the main results of the paper: they formalize and generalize the conclusions of Section \ref{se:examples}.
In Section \ref{se:tails}, we prove the key technical results on the (conditional) tails of marginal distributions of It\^{o} processes. 
%In particular, Lemma \ref{le:necessary.marginal.maximum} provides a uniform exponential bound on the conditional tails -- this result is useful in its own right, and, to the best of our knowledge, it is not available in any existing literature.
Sections \ref{se:pf.1}, \ref{se:pf.2} contain the proofs of the main results.
%In Section \ref{subse:consistency}, we show how to construct equilibria in the models that are ``not too far away" from the model considered in Section \ref{se:examples}, and so that the resulting LOB possesses additional desirable properties.
We conclude in Section \ref{se:conclusion}.

\section{Modeling framework for a finite-frequency auction-style exchange}
\label{se:setup}

\subsection{Mechanics of the exchange}
\label{se:setup}

We consider an exchange in which trading can only occur at discrete times $n=0,1,\ldots,N$.
We assume that the market participants are split into two groups: the \emph{external investors}, who are ``impatient", in the sense that they only submit market orders, and the \emph{strategic players}, who can submit both market and limit orders, and who are willing to optimize their actions over a given (short) time horizon, in order to get a better execution price.\footnote{We do not distinguish the ``aggressive" limit orders, which are posted at the price level of an opposite limit order, and treat them as market orders. This causes no loss of generality, as the market participants in our setting have a perfect observation of the LOB.} In our study, we focus on the strategic players, who are referred to as \emph{agents}, and we model the behavior of the external investors exogenously, via an \emph{exogenous demand}. The interpretation of the external investors is clear: these are the investors who either have a longer-term view on the market, or who simply need to buy or sell the asset for reasons other than short-term profits. The strategic players (i.e., agents), on the contrary, are short-term traders, who attempt to maximize their objective at a shorter time horizon $N$. During every time period $[n,n+1)$, all the orders coming to the exchange are split into \emph{limit} and \emph{market} orders. The limit orders are collected in the so-called \emph{Limit Order Book (LOB)}, and the market orders form the \emph{demand curve}. At time $n+1$, the market orders in the demand curve are executed against the limit orders in the LOB. Then, this process is repeated in the next time interval. 
In particular, during a time period $[n,n+1)$ (for simplicity, we say ``at time $n$"), an agent is allowed to submit a market order, post a limit buy or sell order, or wait (i.e., do nothing). If a limit order is not executed in a given time period, it costs nothing to cancel or re-position it for the next time period.
%We use the convention that all limit orders expire at the beginning of the next period, so that the limit orders have to be re-posted. 
Notice that our framework does not model the time-priority of limit orders. However, introducing a time-priority would not change the agents' maximum objective value, as the ``tick size" is assumed to be zero (i.e., the set of possible price levels is $\RR$), and, hence, an agent can always achieve a priority by posting her order ``infinitesimally" above or below a given competing order.
Further details on modeling the formation of an LOB and the execution rules are presented below.

The demand curves are modeled exogenously by a random field $D=\left(D_n(p)\right)_{p\in\RR,n=1,\ldots,N}$ on a filtered probability space $\left(\Omega,\mathbb{F}=\left(\mathcal{F}_n\right)_{n=0}^N,\PP\right)$, such that $\mathcal{F}_0$ is a trivial sigma-algebra, completed w.r.t. $\PP$. 
%The random variable $D_n(p)$ denotes the demand for the asset at price $p$ and at all more favorable price levels, in the time period $[n-1,n)$. 
The random variable $D^+_n(p) = \max(D_n(p),0)$ denotes the number of shares of the asset that the external investors and the agents submitting market orders are willing to purchase at or below the price $p$, accumulated over the time period $[n-1,n)$, and $D^-_n(p) = -\min(D_n(p),0)$ denotes the number of shares of the asset that the external investors and the agents submitting market orders are willing to sell at or above the price $p$, in the same time period. 
We assume that $D_n(\cdot)$ is a.s. nonincreasing and measurable w.r.t. $\mathcal{F}_n\otimes\mathcal{B}(\RR)$.
We denote by $\mathbb{A}$ a Borel space of \emph{beliefs}, and, for each $\alpha\in\mathbb{A}$, there exists a \emph{subjective probability measure} $\PP^{\alpha}$ on $\left(\Omega,\mathcal{F}_N\right)$, which is absolutely continuous with resect to $\PP$. 
We assume that, for any $n=0,\ldots,N$ and any $\alpha\in\mathbb{A}$, there exists a regular version of the conditional probability $\PP^{\alpha}$ given $\mathcal{F}_n$, denoted $\PP^{\alpha}_n$.\footnote{This assumption holds, for example, if $\mathcal{F}_N$ is generated by a random element with values in a standard Borel space.} We denote the associated conditional expectations by $\EE^{\alpha}_n$. We also need to assume that, for any $\alpha\in\mathbb{A}$, there exists a modification of the family $\left\{\PP^{\alpha}_n\right\}_{n=0}^N$, which satisfies the \emph{tower property with respect to $\PP$}, in the following sense: for any $n\leq m$ and any r.v. $\xi$, such that $\EE^{\alpha} \xi^+ < \infty$, we have
$$
\EE^{\alpha}_n \EE^{\alpha}_m \xi = \EE^{\alpha}_n \xi,\,\,\,\,\,\,\,\,\,\PP\text{-a.s.}
$$
There exists such a modification, for example, if $\PP^{\alpha}\sim\PP$.
In any market model, for every $\alpha$, we fix such a modification of conditional probabilities (up to a set of $\PP$-measure zero) and assume that all conditional expectations $\left\{\EE^{\alpha}_n\right\}$ are taken under this family of measures.
The \emph{Limit Order Book (LOB)} is given by a pair of adapted processes $\nu=(\nu^+_n,\nu^-_n)_{n=0}^{N}$, such that every $\nu^+_n$ and $\nu^-_n$ is a finite sigma-additive random measure on $\RR$ (w.r.t. $\mathcal{F}_n\otimes\mathcal{B}(\RR)$). Herein, $\nu^+_n$ corresponds to the cumulative limit sell orders, and $\nu^-_n$ corresponds to the cumulative limit buy orders, posted at time $n$.%\footnote{For convenience, we sometimes refer to $\nu_n$ as a ``measure", rather than a ``pair of measures".}
The bid and ask prices at any time $n=0,\ldots,N$ are given by the random variables
$$
p^b_n = \sup \text{supp}(\nu^-_n),
\,\,\,\,\,\,\,\,\,\,\,\,\,\,\,\,p^a_n = \inf \text{supp}(\nu^+_n),
$$
respectively. Notice that these extended random variables are always well defined but may take infinite values.

We define the \emph{state space} of an agent as $\mathbb{S}=\RR\times\mathbb{A}$, where the first component denotes the \emph{inventory} of an agent, and the second component denotes her \emph{beliefs}. Every agent in state $(s,\alpha)$ models the future outcomes using the subjective probability measure $\PP^{\alpha}$.
There are infinitely many agents, and their distribution over the state space is given by the \emph{empirical distribution} process $\mu=(\mu_n)_{n=0}^{N}$, such that every $\mu$ is a finite sigma-additive random measure on $\mathbb{S}$ (w.r.t. $\mathcal{F}_n\otimes\mathcal{B}(\mathbb{S})$). In particular, the total mass of agents in the set $S\subset\mathbb{S}$ at time $n$ is given by $\mu_n(S)$.
%We allow for the possibility that the agents can arrive to and leave the market, hence, for the most of the paper, we treat $\mu$ as being given exogenously. Of course, if one assumes that no new agents arrive to or leave the market, then  Notice that this causes no loss of generality for the main results (cf. Section \ref{se:main}), which provide the necessary properties of the equilibria. The case of endogenous $\mu$ is analyzed in Subsection \ref{subse:consistency}.
%It is worth mentioning how we interpret the inventory levels in the case of infinitely many agents. In fact, one should not think of the inventory level $s$ as the actual number of shares, but rather, 
The inventory level $s$ represents the \emph{number of shares per agent}, held in state $(s,\alpha)$. In particular, the total number of shares held by all agents in the set $S\subset\mathbb{S}$ is given by $\int_{S} s\mu_n(ds,d\alpha)$. 
The interpretation of this definition in a finite-player game is discussed in Remark \ref{rem:finplayer} below.
We refer the reader to \cite{GCarmona} for the general concept of a continuum-player game.
%Let us now discus the dynamics of the state process, or, in other words, the execution rules in the exchange.

\begin{remark}\label{rem:finplayer}
The continuum-player game defined in this section can be related to a finite-player game as follows.
%\footnote{See e.g. \cite{GCarmona}, and the references therein, for a similar connection in a general case.}
%Consider a stochastic basis, with a family of measures $\{\PP^{\alpha}\}$, and a random field $D$, as described in Subsection \ref{se:setup}, and a finite-player game in which the values of $\alpha^{i}$ and $s^{i}_0$ represent the beliefs and the initial inventory levels, respectively, of the $i$th agent present in the market, with $i=1,\ldots,M$. 
Denote by $\mu_0$ the empirical distribution of the agents' states at a given time. Recall that $\mu_0$ is a measure on $\mathbb{S}=\RR\times\mathbb{A}$, and assume that it is a finite linear combination of Dirac measures: $\mu_0 = \frac{1}{M} \sum_{i=1}^M \delta_{(s^{i},\alpha^{i})}$.
%Clearly, $\mu_0$ can be interpreted as an empirical distribution of the agents' states in the proposed continuum-player game. 
In this case, we interpret $s^i$ as the {\bf number of shares per agent} held by the agents in the $i$th group. Let us explain how this notion is related to the actual inventory levels (i.e., the actual numbers of shares held by the agents) in the associated finite-player game. To this end, consider a collection of $M$ agents, whose states are given by their (actual) inventories and beliefs, $(s,\alpha)$, with the current states being $\{(\tilde{s}^i = s^i/M,\alpha^i)\}$. Define the ``unit mass" of agents to be $M$. In this finite-player collection, the mass of agents (measured relative to the unit mass, $M$) at any state $(Ms,\alpha)$ is precisely $\mu_0(\{(s,\alpha)\})$, and their total inventory is $Ms\mu_0(\{(s,\alpha)\})$. The number of shares per agent is, then, defined as the total inventory held by these agents divided by their mass, and it is equal to $Ms$. Choosing $s=\tilde{s}^i$, we conclude that, in the finite-player collection, the number of shares per agent held by the agents at state $(\tilde{s}^i,\alpha^i)$ is given by $M\tilde{s}^i = s^i$, which coincides with our interpretation of $s^i$ in the continuum-player game.
It is also easy to show that an equilibrium in the proposed continuum-player game (defined in the next subsection) produces an approximate equilibrium in the associated finite-player game, when the inventory levels $\{\tilde{s}^i\}$ are small (cf. Subsection 2.3 in the extended version of this paper, \cite{GaydukNadtochiy1})
\end{remark}

As the parameter $\alpha$ does not change over time, the state process of an agent, denoted $(S_n)$, is an adapted $\RR$-valued process, representing her inventory.\footnote{Note that, although $\PP^{\alpha}$ does not change over time, the conditional distribution of the future demand, as perceived by the agent, changes dynamically, according to the new information received.}
The control of every agent is given by a triplet of adapted processes $(p,q,r) = (p_n,q_n,r_n)_{n=0}^{N-1}$ on $\left(\Omega,\mathbb{F}\right)$, with values in $\RR^2\times\left\{0,1\right\}$. The first coordinate, $p_n$, indicates the location of a limit order placed at time $n$, and $q_n$ indicates the size of the order (measured in shares per agent, and with negative values corresponding to buy orders).\footnote{Note each agent is only allowed to place her limit order at a single price level, at any given time. However, this entails no loss of optimality. Indeed, using the Dynamic Programming Principle derived in Appendix A, one can show, by induction, that, in equilibrium, an agent does not benefit from posting multiple limit orders at the same time. As shown in \cite{Schmeidler}, this is typical for a continuum-player game.}
%consider a combination of two strategies $(p,q,r)$ and $(p',q',r')$, which coincide everywhere, except that $p_n\neq p_n'$ at a fixed time $n$. Consider a mixed strategy, in which, at time $n$, the agent posts a limit order of size $\lambda q_n$ at the level $p_n$, and another one, of size $(1-\lambda) q_n$, at $p_n'$. The objective value of the mixed strategy is given by (\ref{eq.intro.Jm.def}), with the state process being a linear combination of the state processes of the individual strategies, weighted with $\lambda$ and $(1-\lambda)$. It is easy to see, from the concavity properties, that the mixed objective value does not exceed the maximum between the individual objective values.}
%even if $p_n$ is viewed as a probability measure (i.e. the distribution of the agent's limit orders at time $n$), the natural analogue of the objective, (\ref{eq.intro.Jm.def}), is s.t. it is linear in $p_n$ and concave in  the last term, for $n=N$, becomes concave in $p_{N-1}$.}
%function of the sizes of all submitted orders, provided they have the same sign. Posting limit orders of different signs, simultaneously, is not optimal because only one type of limit orders (i.e. sell or buy) is exercised in a given time period.}
The last coordinate $r_n$ shows whether the agent submits a market order (if $r_n = 1$) or a limit order (if $r_n=0$).
Assume that an agent posts a limit sell order at a price level $p_n$.
If the demand to buy the asset at this price level, $D^+_{n+1}(p_n)$, exceeds the amount of all limit sell orders posted below $p_n$ at time $n$, then (and only then) the limit sell order of the agent is executed. 
Market orders of the agents are always executed at the bid or ask prices available at the time when the order is submitted. We interpret an internal market order (i.e., the one submitted by an agent) as the decision of an agent to join the external investors, in the given time period.
Summing up the above, we obtain the following dynamics for the state process of an agent, starting with initial inventory $s\in\RR$ at time $m=0,\ldots,N-1$:
\begin{equation*}
S^{(p,q,r)}_m(m,s,\nu) = s,\,\,\,\,\,\, \Delta S^{(p,q,r)}_{n+1}(m,s,\nu) = S^{(p,q,r)}_{n+1}(m,s,\nu) - S^{(p,q,r)}_n(m,s,\nu) 
= -q_n \bone_{\left\{ r_n=1 \right\}}
\end{equation*}
\begin{equation}\label{eq.stateProc.def}
- \bone_{\left\{ r_n=0 \right\}}
\left( q^+_n \bone_{\left\{D^+_{n+1}(p_n) > \nu^+_n((-\infty,p_n))\right\}}
-  q^-_n \bone_{\left\{D^-_{n+1}(p_n) > \nu^-_n((p_n,\infty))\right\}}\right),
\,\,\,\,n=m,\ldots,N-1.
\end{equation}
The above dynamics represent an optimistic view on the execution by the agents. In particular, they imply that all limit orders at the same price level are executed in full, once the demand reaches them: i.e., each agent believes that her limit order will be executed first among all orders at a given price level.
In addition, all agents' market orders are executed at the bid and ask prices: i.e., each agent believes that her market order will be executed first, when the demand curve is cleared against the LOB, at the end of a given time period. 
These assumptions can be partially justified by the fact that the agents' orders are infinitesimal: $q_n$ is measured in shares per agent, and an individual agent has zero mass. 
However, if a non-zero mass of agents submit limit orders at the same price level, or execute market orders, at the same time, then, the above state dynamics may violate the market clearing condition: the total size of executed market orders (both in shares and in dollars) may not coincide with the total size of executed limit orders (at least, as viewed by the agents).
%total change in the cumulative inventory of the agents.
%imply that all limit orders at the same price level are executed at the same time, and that all internal market orders are fully executed at the bid or ask price (at least, as viewed by the agents).
%In fact, in such a case, the market clearance condition may be violate, in the sense that: the total executed demand from external investors may not coincide with the total change in the cumulative inventory of the agents.
Nevertheless, this issue is resolved if, at any time, the mass of agents posting limit orders at the same price level or posting market orders is zero. In other words, $(\nu,p,q,r)$ satisfy, $\PP$-a.s.: $\nu_n$ is continuous, as a measure on $\RR$ (i.e., it has no atoms), and $r_n=0$. Such an equilibrium is constructed in Section 8 of the extended version of this paper, \cite{GaydukNadtochiy1}.
The general definition of a continuum-player game and its connection to a finite-player game can be found, e.g., in \cite{GCarmona} and in the references therein (see also Subsection 2.3 in the extended version of this paper, \cite{GaydukNadtochiy1}).

The modeling framework proposed herein has a close connection to the models of \emph{double auctions}, in the economic literature (cf. \cite{DA.DuZhu}, \cite{DA.Vayanos}). The main difference is in the non-standard design of the auction. Namely, in the proposed setting, the auction participants may choose different styles of trading, i.e., market or limit orders, which generates an ex-post information asymmetry between participants: the limit orders have to be submitted before the demand curve is observed, while the market orders are submitted using complete information about the LOB. This difference is not coincidental -- it is, in fact, crucial for a realistic representation of the risks associated with each order type, and it is at the core of the results established herein. A more detailed discussion of the information structure is provided in the next subsection.

\subsection{Equilibrium}
\label{se:equil.def}

The objective function of an agent, starting at the initial state $(s,\alpha)\in\mathbb{S}$, at any time $m=0,\ldots,N$, and using the control $(p,q,r)$, is given by the $\mathcal{F}_m$-measurable random variable
\begin{equation}\label{eq.intro.Jm.def}
J^{(p,q,r)}(m,s,\alpha,\nu) =
\EE^{\alpha}_m \left[ \left(S^{(p,q,r)}_N(m,s,\nu)\right)^+ p^b_N - \left(S^{(p,q,r)}_N(m,s,\nu)\right)^- p^a_N 
\right.
\end{equation}
$$
\left.
- \sum_{n=m}^{N-1} \left(p_n\bone_{\left\{ r_n = 0\right\}} + p^a_n\bone_{\left\{ r_n = 1, q_n <0\right\}} + p^b_n\bone_{\left\{ r_n = 1, q_n >0\right\}} \right) \Delta S^{(p,q,r)}_{n+1}(m,s,\nu) \right],
$$
where we assume that $0\cdot\infty = 0$.
In the above expression, we assume that, at the final time $n=N$, each agent is forced to liquidate her position at the bid or ask prices available at that time. Alternatively, one can think of it as \emph{marking to market} the residual inventory, right after the last external market order is executed.
%To have a well defined optimization problem for each agent, and to obtain its tractable characterization, we need the following admissibility properties. 

%All constructions, notation and assumptions made in the previous section hold throughout.

\begin{definition}\label{def:admis}
For a given LOB $\nu$, integer $m=0,\ldots,N-1$, and state $(s,\alpha)\in\mathbb{S}$, the triplet of adapted processes $(p,q,r)$ is an {\bf admissible control} if the positive part of the expression inside the expectation in (\ref{eq.intro.Jm.def}) has a finite expectation under $\PP^{\alpha}$.
%$$
%\EE^{\alpha}\left( \left(S^{m,s,(p,q,r)}_N\right)^+ p^b_N - \left(S^{m,s,(p,q,r)}_N\right)^- p^a_N \right)^+<\infty,
%\,\,\,\,\,\,\,\,\,\,\,\,
%\EE^{\alpha}\left( p_n \Delta S^{m,s,(p,q,r)}_n \right)^-<\infty,
%$$
%for any $n=m+1,\ldots,N$, $m=0,\ldots,N-1$, and any $(s,\alpha)\in\mathbb{S}$.
\end{definition}

For a given LOB $\nu$, an initial condition $(m,s,\alpha)$, and a triplet of $\mathbb{F}\times\mathcal{B}(\mathbb{S})$-adapted random fields $(p,q,r)$, we identify the latter (whenever it causes no confusion) with stochastic processes $(p,q,r)$ via:
$$
p_n=p_n\left(S^{(p,q,r)}_n(m,s,\nu),\alpha\right),\,\,\,
q_n=q_n\left(S^{(p,q,r)}_n(m,s,\nu),\alpha\right),\,\,\,
r_n=r_n\left(S^{(p,q,r)}_n(m,s,\nu),\alpha\right),
$$
and the state dynamics (\ref{eq.stateProc.def}), for $n=m,\ldots,N$. This system determines $(p,q,r)$ and $S^{(p,q,r)}$ recursively.

\begin{definition}\label{def:optControl}
For a given LOB $\nu$, we call the triplet of progressively measurable random fields $(p,q,r)$ an {\bf optimal control} if, for any $m=0,\ldots,N$ and any $(s,\alpha)\in\mathbb{S}$, we have:
\begin{itemize}
\item $(p,q,r)$ is admissible,

\item $J^{(p,q,r)}(m,s,\alpha,\nu) \geq J^{(p',q',r')}(m,s,\alpha,\nu)$, $\PP$-a.s., for any admissible control $(p',q',r')$.
\end{itemize}
\end{definition}

In the above, we make the standard simplifying assumptions of continuum-player games: each agent is too small to affect the empirical distribution of cumulative controls (reflected in $\nu$) when she changes her control (cf. \cite{GCarmona}).
Note also that our definition of the optimal control implies that it is time consistent: re-evaluation of the optimality at any future step, using the same terminal criteria, must lead to the same optimal strategy.
Next, we discuss the notion of equilibrium in the proposed game.
First, we notice that, if $p^b_N$ or $p^a_N$ becomes infinite, the agents with positive or negative inventory may face the objective value of ``$-\infty$", for any control they use. In such a case, their optimal controls may be chosen in an arbitrary way, resulting in unrealistic equilibria. To avoid this, we impose the additional regularity condition on $\nu$.

\begin{definition}\label{def:admis.LOB}
A given LOB $\nu$ is admissible if, for any $m=0,\ldots,N-1$ and any $\alpha\in\mathbb{A}$, we have, $\PP$-a.s.:
$$
\EE^{\alpha}_m |p^a_N|\vee|p^b_N| < \infty.
$$
\end{definition}

Let us consider the (stochastic) value function of an agent for a fixed $(m,s,\alpha,\nu)$: 
\begin{equation}\label{eq.gen.Val.randField}
V^{\nu}_m(s,\alpha) = \text{esssup}_{p,q,r} J^{(p,q,r)}\left(m,s,\alpha,\nu\right),
\end{equation}
where the essential supremum is taken under $\PP$, over all admissible controls $(p,q,r)$, and $J^{(p,q,r)}$ is given by (\ref{eq.intro.Jm.def}).
Appendix A shows that, for any admissible $\nu$, $V^{\nu}_m(\cdot,\alpha)$ has a continuous modification under $\PP$, which we refer to as the value function of an agent with beliefs $\alpha$.
Using the Dynamic Programming Principle, Appendix A provides an explicit system of recursive equations that characterize optimal strategies and the value function. In particular, the results of Appendix A (cf. Corollary \ref{cor:piecewiseLin}) yield the following proposition.

\begin{proposition}\label{cor:piecewiseLin.new}
Assume that, for an admissible LOB $\nu$, there exists an optimal control $(\hat{p},\hat{q},\hat{r})$. Then, for any $(s,\alpha)\in\mathbb{S}$, the following holds $\PP$-a.s., for all $n=0,\ldots,N-1$:
$$
V^{\nu}_n(s,\alpha) = s^+ \lambda^a_n(\alpha) - s^- \lambda^b_n(\alpha),
$$ 
with some adapted processes $\lambda^a(\alpha)$ and $\lambda^b(\alpha)$, such that $\lambda^a_N(\alpha) = p^b_N$ and $\lambda^b_N(\alpha) = p^a_N$.
\end{proposition}

The values of $\lambda^a(\alpha)$ and $\lambda^b(\alpha)$ can be interpreted as the \emph{expected execution prices} of the agents with beliefs $\alpha$, who are long and short the asset, respectively.

\begin{definition}\label{def:equil.def}
Consider an empirical distribution process $\mu=(\mu_n)_{n=0}^N$ and a market model, as described in Subsection \ref{se:setup}.
We say that a given LOB process $\nu$ and a control $(p,q,r)$ form an {\bf equilibrium}, if there exists a Borel set $\tilde{\mathbb{A}}\subset \mathbb{A}$, called the {\bf support} of the equilibrium, such that:
\begin{enumerate}
\item $\mu_n\left(\mathbb{\RR}\times\left(\mathbb{A}\setminus \tilde{\mathbb{A}}\right)\right)=0$, $\PP$-a.s., for all $n$,
\item $\nu$ is admissible, and $(p,q,r)$ is an optimal control for $\nu$, on the state space $\tilde{\mathbb{S}}=\RR\times\tilde{\mathbb{A}}$,

\item and, for any $n=0,\ldots,N-1$, we have, $\PP$-a.s.,
\begin{equation}\label{eq.nuplus.fixedpoint.def}
\nu^+_n((-\infty,x]) = \int_{\tilde{\mathbb{S}}} \bone_{\left\{p_n(s,\alpha)\leq x, r_n(s,\alpha)=0\right\}}\, q^+_n(s,\alpha) \mu_n(ds,d\alpha),
\,\,\,\,\,\,\forall\, x\in\RR,
\end{equation}
\begin{equation}\label{eq.numinus.fixedpoint.def}
\nu^-_n((-\infty,x]) = \int_{\tilde{\mathbb{S}}} \bone_{\left\{p_n(s,\alpha)\leq x, r_n(s,\alpha)=0\right\}}\, q^-_n(s,\alpha) \mu_n(ds,d\alpha),
\,\,\,\,\,\,\forall\, x\in\RR.
\end{equation} 
\end{enumerate}
\end{definition}

\begin{remark}
It follows from Proposition \ref{cor:piecewiseLin.new} that, in equilibrium, it is optimal for an agent with zero initial inventory to do nothing. Hence, in equilibrium, roundtrip strategies are impossible.
To allow for roundtrip strategies in equilibrium, one can, e.g., introduce an upper bound on $|q|$ or on the total inventory of an agent (as it is done, e.g., in \cite{MMS.gliq1}). However, we do not believe that such a modification would change the qualitative behavior of market liquidity as a function of trading frequency, which is the main focus of the present paper.
\end{remark}

Notice that, because the optimal controls are required to be time consistent under $\PP$, the above definition, in fact, defines a \emph{sub-game perfect equilibrium}.
It is also worth mentioning that Definition \ref{def:equil.def} defines a \emph{partial equilibrium}, as the empirical distribution process $\mu$ is given exogenously. 
A more traditional version of Nash equilibrium would require $\mu$ to be determined by the initial distribution and the values of the state processes:
\begin{equation}\label{eq.endog.mu}
\mu_n = \mu_0 \circ \left( (s,\alpha)\mapsto \left(S_n^{(p,q,r)}(0,s,\nu),\alpha\right) \right)^{-1},
\end{equation}
which must hold $\PP$-a.s., for all $n=0,\ldots,N$, with $S_n^{(p,q,r)}(0,s,\nu)$ defined via (\ref{eq.stateProc.def}), in addition to the other conditions in Definition \ref{def:equil.def}.
Nevertheless, we choose not to enforce the condition (\ref{eq.endog.mu}) in the definition of equilibrium, in order to allow new agents to enter the game, which, in effect, amounts to modeling $\mu$ exogenously. If one assumes that no new agent arrives to the market, then, the fixed-point condition (\ref{eq.endog.mu}) has to be enforced.
Note also that our interpretation of the demand curve $D_n(\cdot)$ implies that it consists of both the external (i.e., due external investors) and internal (i.e., due to the agents) market orders. Therefore, it may be reasonable to consider an additional consistency condition for an equilibrium. A part of this condition is to ensure that a non-zero mass of agents submit market buy orders only if the fundamental price rises above the ask price (i.e., only if a market buy order is actually executed), and, similarly, a non-zero mass of agents submit market sell orders only if the fundamental price falls below the bid price. We assume that the agents' market orders enter into the demand curve with the highest level of priority: e.g., their market buy orders enter the demand curve at the price level infinitesimally close to, but below, the fundamental price, in order to guarantee that they are the first ones to be executed. Thus, another part of the aforementioned consistency condition is to ensure that the absolute value of the demand curve to the left or to the right of the fundamental price is sufficiently large to account for all internal market orders. Mathematically, such consistency condition can be formulated as follows:
\begin{equation}\label{eq.endog.D.1}
d^b_n:=\mu_n\left( \left\{(s,\alpha)\,:\,q_n(s,\alpha)<0,\,r_n(s,\alpha)=1\right\}\right)>0\,\,
\Rightarrow\,\,p^0_{n+1}>p^a_{n},\,\,\lim_{p\uparrow p^0_{n+1}}D^+_{n+1}(p)\geq d^b_n,
\end{equation}
\begin{equation}\label{eq.endog.D.2}
d^a_n:=\mu_n\left( \left\{(s,\alpha)\,:\,q_n(s,\alpha)>0,\,r_n(s,\alpha)=1\right\}\right)>0\,\,
\Rightarrow\,\,p^0_{n+1}<p^b_{n},\,\,\lim_{p\downarrow p^0_{n+1}}D^-_{n+1}(p)\geq d^a_n.
\end{equation}
The above conditions become redundant if the agents never submit market orders in equilibrium.
Section 8 of the extended version of this paper, \cite{GaydukNadtochiy1}, shows how to construct an equilibrium which satisfies condition (\ref{eq.endog.mu}), and in which the agents never submit market orders (hence, (\ref{eq.endog.D.1}) and (\ref{eq.endog.D.2}) are also satisfied). However, it is important to emphasize that the main results of the present work (cf. Section \ref{se:main}) provide necessary conditions for {\bf all} equilibria: for those satisfying the conditions (\ref{eq.endog.mu}), (\ref{eq.endog.D.1}), (\ref{eq.endog.D.2}) and for the ones that do not.

\begin{remark}
Let us comment on the information structure of the game. In the present setting, all agents observe the same information, given by the filtration $\FF$. 
We consider an open-loop Nash equilibrium, in which the agent's strategy is viewed as an adapted stochastic process (rather than a function of the states and controls of other players), and the definition of optimality is chosen accordingly.
In addition, as $\mu$ is adapted to $\FF$, each agent has complete information about the present and past states of other agents, and their beliefs. 
%Using all the information available at time $n$, the agents adjust their strategies, to reach an equilibrium at time $n$, which produces the LOB $\nu_n$. 
However, as the agents use different (subjective) measures $\{\PP^{\alpha}\}$, their views on the future values of $\mu$  may be different. Of course, it would be more realistic to assume that the agents do not have complete information about each other's current states, but this would make the problem significantly more complicated. 
In the present setting, the agents also have complete information about the current location of the fundamental price. In our follow-up paper, \cite{GaydukNadtochiy2}, we relax this assumption, which allows us to develop a more realistic model for the ``local" behavior of an individual agent. However, such a relaxation does not seem necessary for the questions analyzed herein.

As all agents use the same information, the present article belongs to the strand of literature that attempts to explain microstructure phenomena without information asymmetry (cf. \cite{MMS.g3}, \cite{MMS.g6}, \cite{MMS.g1}, \cite{MMS.g2}). Nevertheless, it is important to mention that information asymmetry arises ex-post, between the market participants submitting market and limit orders. This asymmetry is not due to superior information a priori available to any of the agents. Instead, it stems from the very nature of limit orders, which are ``passive" by design (cf. the discussion on the last paragraph of Subsection \ref{se:setup}). Similar observation is made in \cite{MMS.g3}.
\end{remark}

Next, we need to add another condition to the notion of equilibrium. Notice that equations (\ref{eq.nuplus.fixedpoint.def})--(\ref{eq.numinus.fixedpoint.def}) should serve as the fixed-point constraints that enable one to obtain the optimal controls $(p,q,r)$, along with the LOB $\nu$. However, these equations only hold for $n=0,\ldots,N-1$: indeed, the agents do not need to choose their controls at time $n=N$, as the game is over and their residual inventory is marked to the bid and ask prices. However, the terminal bid and ask prices are determined by the LOB $\nu_N$, which, in turn, can be chosen arbitrarily. To avoid such ambiguity, we impose an additional constraint on the equilibria studied herein.
First, we introduce the notion of a \emph{fundamental price}.

\begin{definition}\label{def:het.p0}
Assume that $\PP$-a.s., for any $n=1,\ldots,N$, there exists a unique $p^0_n$ satisfying $D_n\left(p^0_n\right)=0$. Then, the adapted process $(p^0_n)_{n=1}^N$ is called the {\bf fundamental price process}.
\end{definition}

Whenever the notion of a fundamental price is invoked, we assume that it is well defined.
The intuition behind $p^0$ is clear: it is a price level at which the immediate demand is balanced. However, it is important to stress that we do not assume that the asset can be traded at the fundamental price level. Rather, $p^0$ is a feature of the abstract current demand curve, whereas all actual trading happens on the exchange, against the current LOB. This aspect of our setting differs from many other approaches in the literature.

\begin{definition}\label{def:het.}
Assume that the fundamental price is well defined and denote $\xi_N = p^0_N - p^0_{N-1}$. Then, an equilibrium with LOB $\nu$ is {\bf linear at terminal crossing (LTC)} if 
\begin{equation}\label{eq.LTC.def}
\nu_N = \nu_{N-1}\circ (x\mapsto x+\xi_N)^{-1},\,\,\,\,\,\,\,\,\PP\text{-a.s.}
\end{equation}
\end{definition}

The above definition assumes that the terminal LOB $\nu_N$ is obtained from $\nu_{N-1}$ by a simple shift, with the size of the shift equal to the increment in the fundamental price.
This definition connects the LOB at the terminal time with the demand process, ruling out many unnatural equilibria. 
In particular, the question of existence of an equilibrium becomes non-trivial.
However, the mere existence of an equilibrium is not the main focus of the present work:
the existence results, established herein, are limited to Section \ref{se:examples}, which constructs an LTC equilibrium in a specific Gaussian random walk model (a slightly more general existence result is given in Section 8 of the extended version of this paper, \cite{GaydukNadtochiy1}).
What is central to the present investigation is the observation that the agents may reach an equilibrium in which one side of the LOB becomes empty (as demonstrated by the example of Section \ref{se:examples}). We call such LOB, and the associated equilibrium, \emph{degenerate}.

\begin{definition}
We say that an equilibrium with LOB $\nu$ is {\bf non-degenerate} if $\nu^{+}_n(\RR)>0$ and $\nu^{-}_n(\RR)>0$, for all $n=0,\ldots,N-1$, $\PP$-a.s..
\end{definition}

Intuitively, the degeneracy of the LOB refers to a situation where, with positive probability, one side of the LOB disappears from the market: i.e., $\nu^+_n(\RR)$ or $\nu^-_n(\RR)$ becomes zero. Clearly, this happens when the agents who are supposed to provide liquidity choose to post market orders (i.e. consume liquidity) or wait (neither provide nor consume liquidity). Such a degeneracy can be interpreted as the \emph{endogenous liquidity crisis} -- the one that arises purely from the interaction between the agents, and cannot be justified by any fundamental economic reasons (e.g., the external demand for the asset may still be high, on both sides). Taking an optimistic point of view, we assume that the agents choose a non-degenerate equilibrium, whenever one is available. However, if a non-degenerate equilibrium does not exist, an endogenous liquidity crisis may occur with positive probability. One of the main goals of this paper is to provide insights into the occurrence of an endogenous liquidity crisis and its relation to trading frequency.

\section{Example: a Gaussian random walk model}
\label{se:examples}

In this section, we consider a specific market model for the external demand $D$, to construct a non-degenerate LTC equilibrium.
More importantly, using this model, we illustrate the liquidity effects of trading frequency. The present example, albeit very simplistic, enables us to identify the important changes in the optimal strategies of the agents (and, hence, to the LOB) as the trading frequency increases.
In particular, we demonstrate how the \emph{adverse selection} effect may be amplified disproportionally by the high trading frequency and may cause a liquidity crisis. Note that the adverse selection phenomenon, in the present setting, is not a consequence of any ex-ante information asymmetry but is due to the mechanics of the exchange (i.e., the nature of limit orders), which is similar to the phenomena documented in \cite{MMS.g3}, \cite{MMS.g2}.
In the rest of the paper, we show that the conclusions of this section are not due to the particular choice of a model made in the present section and, in fact, persist in a much more general setting.

On a complete stochastic basis $(\Omega,\tilde{\FF}=(\tilde{\mathcal{F}}_t)_{t\in[0,T]},\PP)$, we consider a continuous time process $\tilde{p}_0$:
\begin{equation}\label{eq.p0.BM}
\tilde{p}^0_t = p^0_0 + \alpha t + \sigma W_t,\quad p^0_0 \in \RR, \quad t\in[0,T],
\end{equation} 
where $\alpha\in\RR$ and $\sigma>0$ are constants, and $W$ is a Brownian motion. 
We also consider an arbitrary progressively measurable random field $(\tilde{D}_t(p))$, s.t., $\PP$-a.s., the function $\tilde{D}_t(\cdot)-\tilde{D}_s(\cdot)$ is strictly decreasing and vanishing at zero, for any $0\leq s < t\leq T$.
Finally, we introduce the empirical distribution process $(\tilde{\mu}_t)$, with values in the space of finite sigma-additive measures on $\mathbb{S}$.
We partition the time interval $[0,T]$ into $N$ subintervals of size $\Delta t=T/N$.
A discrete time model is obtained by discretizing the continuous time one\footnote{In order to ensure the existence of regular conditional probabilities for the discrete time model, we can, for example, assume that $\tilde{\mathcal{F}}_T$ is generated by a random element with values in a standard Borel space.}
$$
\mathcal{F}_n = \tilde{\mathcal{F}}_{n\Delta t},\quad p^0_n = \tilde{p}^0_{n\Delta t},
\quad D_n(p) = (\tilde{D}_{n\Delta t}-\tilde{D}_{(n-1)\Delta t})(p-p^0_n),\quad \mu_n = \tilde{\mu}_{n\Delta t}.
$$
In this section, for simplicity, we assume that the set of agents' beliefs is a singleton: $\mathbb{A}=\left\{\alpha\right\}$ and $\PP^{\alpha}=\PP$.
We also assume that (at least, from the agents' point of view) there are always some long and short agents present in the market: $\mu_n\left((0,\infty)\times\mathbb{A}\right),\mu_n\left((-\infty,0)\times\mathbb{A}\right)>0$, $\PP$-a.s., for all $n$.
Clearly, $N$ represents the trading frequency, and the continuous time model represents the ``limiting model," which the agents use as a benchmark, in order to make consistent predictions in the markets with different trading frequencies. We assume that the benchmark model is fixed, and $N$ is allowed to vary.
In the remainder of this section, we propose a method for constructing a non-degenerate LTC equilibrium in the above discrete time model. We show that the method succeeds for any $(N,\sigma)$ if $\alpha=0$. However, for $\alpha\neq 0$, we demonstrate numerically that the method fails as $N$ becomes large enough. We show why, precisely, the proposed construction fails, providing an economic interpretation of this phenomenon. Moreover, we analyze the market close to the moment when a non-degenerate equilibrium fails to exist and demonstrate that the agents' behavior at this time follows the pattern typical for an endogenous liquidity crisis.

In view of Proposition \ref{cor:piecewiseLin.new}, in order to construct a non-degenerate LTC equilibrium, we need to find a control $(\hat{p},\hat{q},\hat{r})$, and the expected execution prices $(\hat{\lambda}^a,\hat{\lambda}^b)$, s.t. the value function of an agent with inventory $s$ is given by $V_n(s) = s^+\hat{\lambda}^a_n - s^-\hat{\lambda}^b_n$, and it is attained by the strategy $(\hat{p},\hat{q},\hat{r})$. In addition, we need to  find a non-degenerate LOB $\nu$, s.t.  (\ref{eq.nuplus.fixedpoint.def}), (\ref{eq.numinus.fixedpoint.def}) and (\ref{eq.LTC.def}) hold.
Our ansatz is as follows
$$
\nu_n = \left(h^a_n \delta_{p^a_n}, h^b_n \delta_{p^b_n}\right),
\quad p^a_n  = \hat{p}^a_n + p^0_n,
\quad p^b_n = \hat{p}^b_n + p^0_n,
\quad -\infty < \hat{p}^b_n,\, \hat{p}^a_n<\infty,
$$
$$
\hat{p}_n(s) = p^a_n\bone_{\{s>0\}} + p^b_n\bone_{\{s<0\}},
\quad \hat{q}_n(s) = s,
\quad \hat{r}_n(s) = 0,
\quad \lambda^a_n = \hat{\lambda}^a_n + p^0_n,
\quad \lambda^b_n = \hat{\lambda}^b_n + p^0_n,
$$
where $\delta$ is the Dirac measure, $(\hat{p}^a,\hat{p}^b,\hat{\lambda}^a,\hat{\lambda}^b)$ are deterministic processes, and $h^a_n=\int_0^{\infty}s\mu_n(ds)>0$, $h^b_n=\int_{-\infty}^0|s|\mu_n(ds)>0$.
With such an ansatz, the conditions (\ref{eq.nuplus.fixedpoint.def}), (\ref{eq.numinus.fixedpoint.def}) are satisfied automatically.
Thus, we only need to choose finite deterministic processes $(\hat{p}^a,\hat{p}^b,\hat{\lambda}^a,\hat{\lambda}^b)$ s.t.: $\hat{p}^a_N = \hat{p}^a_{N-1}$, $\hat{p}^b_N = \hat{p}^b_{N-1}$ (so that the equilibrium is LTC) and the associated $(\hat{p},\hat{q},0)$ form an optimal control, producing the value function $V_n(s) = s^+\lambda^a_n - s^-\lambda^b_n$.
Appendix A contains necessary and sufficient conditions for characterizing such families $(p^a,p^b,\lambda^a,\lambda^b)$.
In particular, we deduce from Corollaries \ref{cor:piecewiseLin} and \ref{cor:piecewiseLin.verif} that $(\hat{p}^a_{N-1},\hat{p}^b_{N-1},\hat{\lambda}^a_{N-1},\hat{\lambda}^b_{N-1})$ form a suitable family in a single-period case, $[N-1,N]$, if they solve the following system:
\begin{equation}\label{eq.ex.singleStep.1}
\left\{
\begin{array}{l}
{\hat{p}^a_{N-1} \in \text{arg}\max_{p\in\RR} \EE\left((p - \hat{p}^b_{N-1} - \xi) \bone_{\left\{ \xi > p \right\}}\right),
\quad \hat{p}^b_{N-1}<0,\phantom{\frac{\frac{1}{2}}{2}}}\\
{\hat{p}^b_{N-1} \in \text{arg}\max_{p\in\RR} \EE\left((\hat{p}^a_{N-1} - p + \xi) \bone_{\left\{ \xi < p \right\}}\right),
\quad \hat{p}^a_{N-1}>0,\phantom{\frac{\frac{1}{2}}{2}}}\\
{\hat{\lambda}^a_{N-1} = \hat{p}^b_{N-1} + \alpha\Delta t + \EE\left((\hat{p}^a_{N-1} - \hat{p}^b_{N-1} - \xi) \bone_{\left\{ \xi > \hat{p}^a_{N-1} \right\}}\right), \phantom{\frac{\frac{\frac{1}{2}}{2}}{2}}}\\
{\hat{\lambda}^b_{N-1} = \hat{p}^a_{N-1} + \alpha\Delta t - \EE\left((\hat{p}^a_{N-1} - \hat{p}^b_{N-1} + \xi) \bone_{\left\{ \xi < \hat{p}^b_{N-1} \right\}}\right), \phantom{\frac{\frac{\frac{1}{2}}{2}}{2}}}\\
{\hat{p}^b_{N-1}\leq \hat{\lambda}^a_{N-1},
\quad \hat{\lambda}^b_{N-1} \leq \hat{p}^a_{N-1},
\quad \hat{p}^a_{N-1} \geq \hat{p}^b_{N-1} + |\alpha|\Delta t, \phantom{\frac{\frac{1}{2}}{2}}}
\end{array}
\right.
\end{equation}
where $\xi=\Delta p^0_N\sim \mathcal{N}(\alpha\Delta t, \sigma^2 \Delta t)$. Let us comment on the economic meaning of the equations in (\ref{eq.ex.singleStep.1}).
The expectations in the first two lines represent the \emph{relative expected profit} from executing a limit order at time $N$, at the chosen price level $p+p^0_{N-1}$, versus marking the inventory to market at time $N$, at the best price available on the other side of the book: i.e., $p^b_N = \hat{p}^b_{N-1} + \xi + p^0_{N-1}$ or $p^a_N = \hat{p}^a_{N-1} + \xi + p^0_{N-1}$.
Notice that a limit order is executed if and only if the fundamental price at time $N$ is above or below the chosen limit order: i.e., if $p^0_{N-1} + \xi > p + p^0_{N-1}$ or $p^0_{N-1} + \xi < p + p^0_{N-1}$.\footnote{The execution of limit orders simplifies in the chosen ansatz, because the agents on each side of the book (i.e., long or short) post orders at the same prices.}
%This profit is measured relative to the expected execution price in the next period: $\lambda^a_{N} + p^0_N = p^b_{N-1} + \xi + p^0_{N-1}$ and $\lambda^b_N + p^0_N = p^a_{N-1} + \xi + p^0_{N-1}$. 
Clearly, it is only optimal for an agent to post a limit order if the relative expected profit is nonnegative, which is the case if and only if $\hat{p}^b_{N-1}<0<\hat{p}^a_{N-1}$. 
The third and fourth lines in (\ref{eq.ex.singleStep.1}) represent the expected execution prices of the agents at time $N-1$, assuming they use the controls given by $(\hat{p}^a_{N-1},\hat{p}^b_{N-1})$. Each of the right hand sides is a sum of two components: the relative expected profit from posting a limit order and the expected value of marking to market at time $N$, measured relative to $p^0_{N-1}$.
Let us analyze the inequalities in the last line of (\ref{eq.ex.singleStep.1}).
%The last line ensures that no abnormalities happen, and a non-degenerate LTC equilibrium exists, in the form stated in the ansatz.
If the bid price at time $N-1$ exceeds the expected execution price of a long agent, i.e., $\hat{p}^b_{N-1} + p^0_{N-1}> \hat{\lambda}^a_{N-1}+ p^0_{N-1}$, then every agent with positive inventory prefers to submit a market order, rather than a limit order, at time $N-1$, which causes the ask side of the LOB to degenerate. Similarly, we establish $\hat{\lambda}^b_{N-1} \leq \hat{p}^a_{N-1}$.
%If $p^a_{N-1}<0$, an agent may post a limit sell order a time $N-1$ at a price level $p+p^0_{N-1}$, and, if the order gets executed and $\xi\approx p$, the agent will buy the shares back at time $N$, at the price $p^a_N + p^0_N = p^a_{N-1} + p^0_{N-1} + \xi \approx p^a_{N-1} + p^0_{N-1} + p < p+p^0_{N-1}$. As the agent is small and does not assume that she, alone, can affect the LOB, she can scale up this strategy to make infinite expected profits. The latter contradicts the existence of equilibrium and, hence, is excluded by the conditions $3$--$4$ in Corollary \ref{cor:piecewiseLin}.\footnote{The inequalities $p^b_{N-1}\leq 0\leq p^a_{N-1}$, in general, are not necessary for the conditions $3$--$4$ of Corollary \ref{cor:piecewiseLin} to hold, but they are sufficient.}
Finally, if $\alpha>0$ and $\hat{p}^a_{N-1} < \hat{p}^b_{N-1} + \alpha\Delta t$, an agent may buy the asset using a market order at time $N-1$, at the price $\hat{p}^a_{N-1}+p^0_{N-1}$, and sell it at time $N$, at the expected price $\hat{p}^b_{N-1} + p^0_{N-1} + \alpha \Delta t > \hat{p}^a_{N-1}+p^0_{N-1}$ (a reverse strategy works for $\alpha<0$). This strategy can be scaled to generate infinite expected profit and, hence, is excluded by the last inequality in the last line of (\ref{eq.ex.singleStep.1}).

We construct a solution to (\ref{eq.ex.singleStep.1}) by solving a fixed-point problem given by the first two lines of (\ref{eq.ex.singleStep.1}) and verifying that the desired inequalities hold.\footnote{In fact, it is not difficult to prove rigorously that, for any $(\alpha,\sigma)$, there exists a unique solution to such a system, provided $\Delta t$ is small enough. We omit this result for the sake of brevity.}
%following simplified system
%\begin{equation}\label{eq.ex.singleStep.2}
%\left\{
%\begin{array}{l}
%{p^a_{N-1} \in \text{arg}\max_{p\in\RR} \EE\left[(p - p^b_{N-1} - \xi) \bone_{\left\{ \xi > p \right\}}\right],\phantom{\frac{\frac{1}{2}}{2}}}\\
%{p^b_{N-1} \in \text{arg}\max_{p\in\RR} \EE\left[(p^a_{N-1} + \xi - p) \bone_{\left\{ \xi < p \right\}}\right], \phantom{\frac{\frac{1}{2}}{2}}}\\
%{p^b_{N-1} < 0 < p^a_{N-1}.\phantom{\frac{\frac{1}{2}}{2}}}\\
%\end{array}
%\right.
%\end{equation}
%If $\alpha=0$, due to symmetry, one can easily reduce the system to a one-dimensional fixed-point problem, by assuming $p^a_{N-1} = -p^b_{N-1}$.
%For an arbitrary $\alpha$, the system consisting of the first two lines of (\ref{eq.ex.singleStep.1}) can be formulated as a fixed-point problem and solved numerically.
We implement this computation in MatLab, and the results can be seen as the right-most points on the graphs in Figure \ref{fig:2}.
From the numerical solution, we see that, whenever $\Delta t$ is small enough, the conditions $\hat{p}^b_{N-1}\leq \hat{\lambda}^a_{N-1}$ and $\hat{\lambda}^b_{N-1} \leq \hat{p}^a_{N-1}$ are satisfied (cf. the right part of Figure \ref{fig:2}).\footnote{This is easy to explain intuitively, as the optimal objective values in the first two lines of (\ref{eq.ex.singleStep.1}) are of the form $C\sqrt{\Delta t} + \alpha \underline{\underline{O}}(\Delta t)$.}
%In fact, it is not difficult to prove that, for any $(\alpha,\sigma)$ there exists a unique solution to this system, provided $\Delta t$ is small enough (we omit this result for the sake of brevity). In addition, it can be shown that, for a fixed $\sigma$, the solution $(p^a_{N-1},p^b_{N-1})$ is such that 
%$$
%\EE\left[(p^a_{N-1} - p^b_{N-1} - \xi) \bone_{\left\{ \xi > p^a_{N-1} \right\}}\right] = C \sqrt{\Delta t} + f(\alpha) O(\Delta t)
%$$ 
%holds for all small enough $\Delta t$, where $C$ is a constant and $O(\Delta t)$ depends on $(\alpha,\Delta t)$, and is uniformly bounded by a constant times $\Delta t$, whenever $\alpha$ changes over a compact.
%For any $p^b_{N-1}<0$, using the properties of Gaussian distribution, it is easy to deduce that the function $p\mapsto \EE\left[(p - p^b_{N-1} - \xi) \bone_{\left\{ \xi > p \right\}}\right]$ has a single global maximum and that it vanishes as $p\rightarrow\infty$. Hence, there exists $p\in\RR$ at which this function takes a strictly positive value. 
In addition, for $\alpha\geq0$, we have
$$
0 < \EE\left(\hat{p}^a_{N-1} - \hat{p}^b_{N-1} - \xi\,\vert\,\xi > \hat{p}^a_{N-1} \right)
= \hat{p}^a_{N-1} - \hat{p}^b_{N-1} - \EE\left(\xi\,\vert\,\xi > \hat{p}^a_{N-1} \right)
\leq \hat{p}^a_{N-1} - \hat{p}^b_{N-1} - \alpha\Delta t,
$$
which yields the last inequality in (\ref{eq.ex.singleStep.1}). The case of $\alpha<0$ is treated similarly.
%Notice that the above construction yields a non-degenerate equilibrium in any single-step model, regardless of the drift, as long as $\Delta t>0$ is small enough. However, we do know from Theorem \ref{thm:main.necessary} that, when $\alpha\neq0$, this construction has to fail. Clearly, this has to happen during the recursive procedure -- going from $n=N-1$ to $n=0$. Let us illustrate what exactly fails during this recursive procedure. 
Notice that $\hat{\lambda}^a_N = \hat{p}^b_N = \hat{p}^b_{N-1}$ and $\hat{p}^a_{N-1} = \hat{p}^a_N = \hat{\lambda}^b_N$. Thus, the single-period equilibrium we have constructed satisfies:
\begin{equation}\label{eq.example.1period.ineq.1}
\hat{p}^b_{n}\leq \hat{\lambda}^a_{n},\quad \hat{\lambda}^b_{n} \leq \hat{p}^a_{n},
\quad \hat{\lambda}^a_{n+1} < 0,\quad \hat{\lambda}^b_{n+1}>0,
\end{equation}
for $n=N-1$.
If one of the first two inequalities in (\ref{eq.example.1period.ineq.1}) fails, the agents choose to submit market orders, as opposed to limit orders, which leads to \emph{degeneracy} of the LOB -- one side of it disappears. 
If one of the last two inequalities fails, the execution of a limit order, at any price level, yields a negative relative expected profit for the agents on one side of the book (given by the expectation in the first or second line of (\ref{eq.ex.singleStep.1})). As a result, it becomes optimal for all such agents to stop posting any limit orders, and the LOB degenerates. The latter is interpreted as the \emph{adverse selection} effect. For example, if the third inequality in (\ref{eq.example.1period.ineq.1}) fails, then, every long agent believes that, no matter the price at which her limit order is posted, if it is executed in the next time period, her expected execution price at the next time step will be higher than the price at which the limit order is executed. Hence, it suboptimal to post a limit order at all.
%The role of the last two inequalities is more subtle. If one of them fails, so that an opposite inequality holds, then, as discussed in the paragraph following (\ref{eq.ex.singleStep.1}), the agents may generate infinite expected profits, which means that an LTC equilibrium does not exist. However, if one of the inequalities turns into an equality, an LTC equilibrium may still exist, but the objective function in one of the first two lines of (\ref{eq.ex.singleStep.1}) becomes strictly negative. The latter means that the execution of a limit order, at any price level, yields a negative relative expected profit for the agent (measured relative to the alternative of doing nothing). This is interpreted as the \emph{adverse selection} effect: the agents believe that, no matter at which price level a limit order is posted, if it is executed, the price will tend to go in the same direction afterwards, hence, it does not make sense to post a limit order at all.

In a single period $[N-1,N]$, by choosing small enough $\Delta t$, we can ensure that the inequalities in (\ref{eq.example.1period.ineq.1}) are satisfied. However, it turns out that, as we progress recursively, constructing an equilibrium, we may encounter a time step at which one of the inequalities in (\ref{eq.example.1period.ineq.1}) fails, implying that a non-degenerate LTC equilibrium cannot be constructed for the given time period (at least, using the proposed method).
%The above inequality can be interpreted as follows: the agents trying to sell the asset expect to sell it at a discount relative to the fundamental price. Similarly, the agents trying to buy the asset expect to pay more than the fundamental price. This may raise the following question: why do the agents even participate in the game, if they are better off trading at the fundamental price? However, this seemingly contradictory behavior can be explained by the fact that, in our framework, the fundamental price does not have a meaning of an actual price level at which the transactions occur. 
%Rather, it is an abstract quantity that represents the price level at which the demand is balanced. 
%Recall that the external investors trying to purchase the asset have to pay at least the ask price, which is above the fundamental price, and, typically, above $\lambda^b$. Thus, the strategic players (aka agents) still make profits relative to the external investors.
%Notice that we have also discovered another reason to view the values of $|\lambda^a|$ and $|\lambda^b|$ as the measures of ``market inefficiency": they can be viewed as the cost to ``discover" the fundamental price $p^0$.
%It turns out that the signs of $\lambda^a$ and $\lambda^b$ are exactly what determines whether the market degenerates or not! 
To see this, consider the recursive equations for $(\hat{p}^a,\hat{\lambda}^a)$ (which are chosen to satisfy the conditions of Corollary \ref{cor:piecewiseLin}, in Appendix A, given our ansatz):
\begin{equation}\label{eq.RW.pan}
\left\{
\begin{array}{l}
{\hat{p}^a_n \in \text{arg}\max_{p\in\RR} \EE \left(\left(p-\hat{\lambda}^a_{n+1} - \xi\right) \bone_{\left\{ \xi> p \right\}}\right),
\phantom{\frac{\frac{1}{2}}{2}}}\\
{\hat{\lambda}^a_n = \hat{\lambda}^a_{n+1} + \alpha\Delta t + \EE \left( \left(\hat{p}^a_n - \hat{\lambda}^a_{n+1} - \xi\right) \bone_{\left\{\xi> \hat{p}^a_n ) \right\}} \right) <0,\phantom{\frac{\frac{1}{2}}{2}}}
\end{array}
\right.
\end{equation}
and similarly for $(\hat{p}^b,\hat{\lambda}^b)$.
%where we assume that $p^a_{n}=\infty$ and $\lambda^a_{n} =  \max\left(p^b_n,\lambda^a_{n+1} + \alpha\Delta t\right)$, if the associated ``$\text{arg}\max$" is empty.
%In addition, to satisfy the conditions $2$, $4$ of Corollary \ref{cor:piecewiseLin}, we require that: $\sup_{p}\EE\left( \lambda^a_{n+1} - p + \xi \,|\, \xi<p \right)\leq 0$ and $p^a_n \geq \lambda^a_{n+1} + \alpha\Delta t$. If $\lambda^a_{n+1}\leq0$, the first inequality is, clearly, satisfied. To show that, in this case, the second inequality also holds, we notice that, if $\lambda^a_{n+1}=0$, then $p^a_n=\infty$, as the associated ``$\text{arg}\max$" is empty. If $\lambda^a_{n+1}<0$, then, as discussed earlier, the ``$\max$" in the first line of (\ref{eq.RW.pan}) is strictly positive, and 
%$$
%p^a_n \geq \lambda^a_{n+1} + \EE \left[ \xi\,\vert\,\xi> p^a_n \right] \geq \lambda^a_{n+1} + \alpha\Delta t.
%$$
Using the properties of the Gaussian distribution, it is easy to see that, if $\hat{\lambda}^a_{n+1}<0$, we have $\hat{p}^a_n>0$.
Similar conclusion holds for $(\hat{\lambda}^b,\hat{p}^b)$.
Thus, if $\hat{\lambda}^a_k< 0 < \hat{\lambda}^b_k$, for $k=n+1,\ldots,N$, our method allows us to construct a non-degenerate LTC equilibrium on the time interval $[n,N]$, with $\hat{p}^b<0<\hat{p}^a$.
%If the latter inequalities for $(\lambda^a_k,\lambda^b_k)$ are strict, then, this equilibrium is non-degenerate.
Such a construction always succeeds if the agents are market-neutral: i.e., $\alpha=0$. Indeed, in this case, assuming $\hat{\lambda}^a_{n+1} < 0 < \hat{\lambda}^b_{n+1}$, we have $\hat{p}^b_n < 0 < \hat{p}^a_n$ and
$$
\hat{\lambda}^a_{n+1} + \left(\EE \left( \left(\hat{p}^a_n - \hat{\lambda}^a_{n+1} - \xi\right) \bone_{\left\{\xi> \hat{p}^a_n ) \right\}} \right)\right)^+
= \EE\left( \hat{\lambda}^a_{n+1} \bone_{\left\{\xi> \hat{p}^a_n ) \right\}} \right)
+ \EE \left( \left(\hat{p}^a_n - \xi\right) \bone_{\left\{\xi> \hat{p}^a_n ) \right\}} \right)
< 0.
$$
Hence, $\hat{\lambda}^a_n < 0$, and, similarly, we deduce that $\hat{\lambda}^b_n>0$. By induction, we obtain a non-degenerate LTC equilibrium on $[0,N]$, for any $(N,\sigma)$, as long as $\alpha=0$.
%The resulting $(p^a,p^b,\lambda^a,\lambda^b)$ are plotted in Figure \ref{fig:1}.
Corollary \ref{prop:main.smallspread} shows that, as $N\rightarrow\infty$, the processes $(\hat{\lambda}^a,\hat{\lambda}^b)$ converge to zero, which means that the expected execution prices converge to the fundamental price. The latter is interpreted as \emph{market efficiency} in the high-frequency trading regime: any market participant expects to buy or sell the asset at the fundamental price. The left hand side of Figure \ref{fig:3} shows that the bid and ask prices also converge to the fundamental price if $\alpha=0$. This can be interpreted as a \emph{positive liquidity effect} of increasing the trading frequency.

However, the situation is quite different if $\alpha\neq 0$. Assume, for example, that $\alpha>0$. Then, the second line of (\ref{eq.RW.pan}) implies that $\hat{\lambda}^a$ increases by, at least, $\alpha\Delta t$ at each step of the (backward) recursion. Recall that the number of steps is $N=T/\Delta t$, hence, $\hat{\lambda}^a_0 \geq \hat{\lambda}^a_N + \alpha T$. If $|\hat{\lambda}^a_N|$ is small (which is typically the case if $N$ is large), then, we may obtain $\hat{\lambda}^a_{n+1}\geq 0$, at some time $n$, which violates the third inequality in (\ref{eq.example.1period.ineq.1}), or, equivalently, implies that the objective in the first line of (\ref{eq.RW.pan}) is strictly negative for all $p$. The latter implies that it is suboptimal for the agents with positive inventory to post limit orders, and the proposed method fails to produce a non-degenerate LTC equilibrium in the interval $[n,N]$. Figure \ref{fig:2} shows that this does, indeed, occur. Figures \ref{fig:2} and \ref{fig:3} also show that, for a given (finite) frequency $N$, if $|\alpha|$ is small enough, a non-degenerate equilibrium may still be constructed. Nevertheless, for any $|\alpha|\neq0$, however small it is, there exists a large enough $N$, s.t. the non-degenerate LTC equilibrium fails to exist (at least, within the class defined by the proposed method). This is illustrated in Figure \ref{fig:3}.

It is important to provide an economic interpretation of why such degeneracy occurs. A careful examination of Figure \ref{fig:2} reveals that, around the time when $\hat{\lambda}^a$ becomes nonnegative, the ask price $\hat{p}^a$ explodes. This means that the agents who want to sell the asset are only willing to sell it at a very high price. Notice also that this price is several magnitudes larger than the expected change in the fundamental price (represented by the black dashed line in the left hand side of Figure \ref{fig:2}). Hence, such a behavior cannot be justified by the behavior of the fundamental. Indeed, this is precisely what is called an \emph{endogenous liquidity crisis}. So, what causes such a liquidity crisis? Recall that there are two potential reasons for the market to degenerate: agents may choose to submit market orders (if $\hat{p}^b_n>\hat{\lambda}^a_n$ or $\hat{p}^a_n<\hat{\lambda}^b_n$), or they may choose to wait and do nothing (if $\hat{\lambda}^a_{n+1}\geq 0$ or $\hat{\lambda}^b_{n+1}\leq 0$). The right hand side of Figure \ref{fig:2} shows that the degeneracy is caused by the second scenario. This means that the naive explanation of an endogenous liquidity crisis, based on the claim that, in a bullish market, those who need to buy the asset will submit market orders wiping out liquidity on the sell side of the book, is wrong. Instead, if the agents on the sell side of the book have the same beliefs, they will increase the ask price so that it is no longer profitable for the agents who want to buy the asset to submit market buy orders. In fact, the ask price may increase disproportionally to the expected change in the fundamental price (i.e., the signal), and this is what causes an endogenous liquidity crisis. The size of the resulting change in the bid or ask price depends not only on the signal, but also on the trading frequency, which demonstrates the \emph{negative liquidity effect} of increasing the trading frequency: it fragilizes the market with respect to deviations of the agents from market-neutrality.
The latter, in turn, is explained by the fact that higher trading frequency exacerbates the \emph{adverse selection} effect. To see this, consider, e.g., an agent who is trying to sell one share of the asset. Increasing the trading frequency increases the expected execution value of this agent, bringing it closer to the fundamental price: this corresponds to $\hat{\lambda}^a$ approaching zero (from below). 
Assume that the agent posts a limit sell order at a price level $p$. If this order is executed in the next period, then, the agent receives $p$, but, for this to happen, the fundamental price value at the next time step, $p^0_{n+1}$, has to be above $p$. On the other hand, the expected execution price of the agent at the next time step is $p^0_{n+1} + \hat{\lambda}^a_{n+1}$.
Thus, the expected relative profit, given the execution of her limit order, is $\EE_n (p - p^0_{n+1} - \hat{\lambda}^a_{n+1} \, |\,p^0_{n+1}>p )$. The latter expression cannot be positive, unless $\hat{\lambda}^a_{n+1}<0$ and $|\hat{\lambda}^a_{n+1}|$ is sufficiently large.
Therefore, if $|\hat{\lambda}^a_{n+1}|$ is small relative to $\EE_n (p^0_{n+1} - p\, |\,p^0_{n+1}>p)$, the
%expected future execution value of one share (as viewed by the agent at time $n$, and conditional on the fundamental price moving above $p$) may become larger than $p$. As a result, the 
agent is reluctant to post a limit order at the price level $p$. Hence, $p$ needs to be sufficiently large, to ensure that $\EE_n (p^0_{n+1} - p\, |\,p^0_{n+1}>p )$ is smaller than $|\hat{\lambda}^a_{n+1}|$ (in the Gaussian model of this section, the latter expectation vanishes as $p\rightarrow\infty$) -- and the extent to which $p$ needs to increase determines the effect of adverse selection.
It turns out that, if the agents are market-neutral (i.e. $\alpha=0$), as the frequency $N$ increases, the quantity $\EE_n (p^0_{n+1} - p\, |\,p^0_{n+1}>p )$, for any fixed $p$, converges to zero at the same rate as $|\hat{\lambda}^a_{n+1}|$, hence, the above adverse selection effect does not get amplified. On the contrary, if the agents are not market-neutral, $\hat{\lambda}^a_{n+1}$ reaches zero at some high enough (but finite) frequency, while $\EE_n (p^0_{n+1} - p\, |\,p^0_{n+1}>p )$ remains strictly positive, for any finite $p$, which produces an ``infinite" adverse selection effect and causes the market to degenerate.
Of course, so far, these conclusions are based on a very specific example and on a particular method of constructing an equilibrium. The next section shows that they remain valid in any model (with, possibly, heterogeneous beliefs) in which the fundamental price is given by an It{\^o} process. 

It is worth mentioning that a similar adverse selection effect arises in \cite{MMS.g3}, and it is referred to as the ``winner's curse" in \cite{MMS.g2}. However, the latter papers do not investigate the nature of this phenomenon and focus on other questions instead. In the literature on double auctions (cf. \cite{DA.DuZhu}, \cite{DA.Vayanos}), a similar effect arises when the auction participants choose to decrease their trading activity in a given auction, because they expect many more opportunities to trade in the future. The latter is similar to the agents choosing to forgo limit orders and wait, in the present example.

\section{Main results}
\label{se:main}

In this section, we generalize the previous conclusions, so that they hold in a general model and for any equilibrium. As before, we begin with the ``limiting" continuous time model.
Consider a terminal time horizon $T>0$ and a complete stochastic basis $(\Omega,\tilde{\FF}=(\tilde{\mathcal{F}}_t)_{t\in[0,T]},\PP)$, with a Brownian motion $W$ on it.\footnote{In order to ensure the existence of regular conditional probabilities for the discrete time model, we can, for example, assume that $\tilde{\mathcal{F}}_T$ is generated by a random element with values in a standard Borel space.} 
%As usual, we denote by $\PP_t$ the regular conditional probability probability given $\mathcal{F}_t$, under $\PP$. The associated expectations are denoted $\EE_t$.
We define the adapted process $\tilde{p}^0$ as a continuous modification of
\begin{equation}\label{eq.p0.cont}
\tilde{p}^0_t = p^0_0 + \int_0^t \sigma_s dW_s,\,\,\,\,\,\,\,\,\,\,\,p^0_0 \in \RR,
\end{equation}
where $\sigma$ is a progressively measurable locally square integrable process. 
%We make the following assumption on $\sigma$. 

\begin{ass}\label{ass:sigma}
There exists a constant $C>1$, such that, $1/C\leq \sigma_t\leq C$, for all $t\in[0,T]$, $\PP$-a.s..
\end{ass}

Consider a Borel set of beliefs $\mathbb{A}$ and the associated family of measures $\left\{\PP^{\alpha}\right\}_{\alpha\in\mathbb{A}}$ on $(\Omega,\tilde{\mathcal{F}}_T)$, absolutely continuous with respect to $\PP$.
Then, for any $\alpha\in\mathbb{A}$, we have
\begin{equation*}\label{eq.p0.cont.a}
\tilde{p}^0_t = p^0_0 + A^{\alpha}_t + \int_0^t \sigma_s dW^{\alpha}_s,\quad p^0_0 \in \RR,
\quad \PP^{\alpha}\text{-a.s.},\,\,\forall t\in[0,T],
\end{equation*}
where $W^{\alpha}$ is a Brownian motion under $\PP^{\alpha}$, and $A^{\alpha}$ is a process of finite variation.  
%As before, $\PP_t^{\alpha}$ is the regular conditional probability given $\mathbb{F}_t$, under $\PP$, and $\EE_t^{\alpha}$ is the expectation under this measure.
We assume that $A^{\alpha}$ is absolutely continuous: i.e., for any $\alpha\in\mathbb{A}$, there exists a locally integrable process $\mu^{\alpha}$, such that
$$
A^{\alpha}_t = \int_0^t \mu^{\alpha}_s ds,\quad \PP^{\alpha}\text{-a.s.},\,\,\forall t\in[0,T].
$$

\begin{ass}\label{ass:A.alpha}
For any $\alpha\in\mathbb{A}$, the process $\mu^{\alpha}$ is $\PP$-a.s. right-continuous, and there exists a constant $C>0$, such that $|\mu^{\alpha}_t| \leq C$, for all $t\in[0,T]$, $\PP$-a.s..
\end{ass}

Thus, we can rewrite the dynamics of $\tilde{p}^0$, under each $\PP^{\alpha}$, as follows: $\PP^{\alpha}$-a.s., the following holds for all $t\in[0,T]$
\begin{equation}\label{eq.p0.cont.a.alpha}
\tilde{p}^0_t = p^0_0 + \int_0^t \mu^{\alpha}_s ds + \int_0^t \sigma_s dW^{\alpha}_s,\quad p^0_0 \in \RR.
\end{equation}
In addition, we modify the above stochastic integral on a set of $\PP^{\alpha}$-measure zero, so that (\ref{eq.p0.cont.a.alpha}) holds for \emph{all} $(t,\omega)$.
%We denote by $\tilde{\EE}^{\alpha}_t$ the conditional expectation given $\tilde{\mathcal{F}}_t$, under $\PP^{\alpha}$.
In what follows, we often need to analyze the future dynamics of $\tilde{p}^0$ under $\PP^{\alpha}$, conditional on $\tilde{\mathcal{F}}_t$, for various $(t,\alpha)$ simultaneously. This is why we need the following joint regularity assumption.

\begin{ass}\label{ass:joint.cond.reg}
There exists a modification of regular conditional probabilities 
$$
\left\{\tilde{\PP}^{\alpha}_t=\PP^{\alpha}\left(\cdot\,|\,\tilde{\mathcal{F}}_t\right)
%,\PP_t=\PP\left(\cdot\,|\,\mathcal{F}_t\right)\,
\right\}_{t\in[0,T],\,\alpha\in\mathbb{A},}
$$
such that it satisfies the tower property with respect to $\PP$ (as described in Section \ref{se:setup}).
%, and, $\PP$-a.s., for all $\alpha\in\mathbb{A}$ and all $t\in[0,T]$, the future price process $(\tilde{p}^0_s)_{s\in[t,T]}$ satisfies (\ref{eq.p0.cont.a.alpha}) $\tilde{\PP}^{\alpha}_t$-a.s..
%\begin{itemize}
%\item the future price process $(\tilde{p}^0_s)_{s\in[t,T]}$ satisfies (\ref{eq.p0.cont.a}), $\PP^{\alpha}_t$-a.s., 
%\item and $\PP^{\alpha}_t$ is absolutely continuous with respect to $\PP_t$.
%\end{itemize}
\end{ass}

Assumption \ref{ass:joint.cond.reg} is satisfied, for example, if $\PP^{\alpha}\sim \PP$, for all $\alpha\in\mathbb{A}$, or if the set $\mathbb{A}$ is countable.
%and, $\PP$-a.s., (\ref{eq.p0.cont.a.alpha}) holds for all $t\in[0,T]$ and all $\alpha\in\mathbb{A}$.
%The above assumption is satisfied, for example, if $\PP^{\alpha}\sim\PP$, for all $\alpha\in\mathbb{A}$, and the equation (\ref{eq.p0.cont.a.alpha}) holds for \emph{all} random outcomes $\omega$ (with a fixed continuous modification of the stochastic integral, for each $\alpha\in\mathbb{A}$).
%Alternatively, it suffices to require that the set $\mathbb{A}$ is countable.
%Note that, for all $(\omega,\alpha,t)$, the future price process $(\tilde{p}^0_s)_{s\in[t,T]}$ satisfies (\ref{eq.p0.cont.a.alpha}) $\tilde{\PP}^{\alpha}_t$-a.s.
Throughout the rest of the paper, $\tilde{\PP}^{\alpha}_t$ refers to a member of the family appearing in Assumption \ref{ass:joint.cond.reg}. All conditional expectations $\tilde{\EE}^{\alpha}_t$ are taken under such $\tilde{\PP}^{\alpha}_t$.

The main results of this section require additional continuity assumptions on $\sigma$ and $\mu^{\alpha}$. The following assumption can be viewed as a stronger version of $\mathbb{L}^2$-continuity of $\sigma$.

\begin{ass}\label{ass:main.L2.strong}
%Let $\left\{\PP^{\alpha}_{t}\right\}$ be the modification of the family of regular conditional probabilities appearing in Assumption \ref{ass:joint.cond.reg}.
There exists a function $\varepsilon(\cdot)\geq0$, such that $\varepsilon(\Delta t)\rightarrow0$, as $\Delta t\rightarrow0$, and, $\PP$-a.s., 
$$
\tilde{\PP}^{\alpha}_{t} \left(\EE^{\alpha}\left(\left(\sigma_{s\vee\tau} - \sigma_{\tau} \right)^2 \,|\,\mathcal{F}_{\tau} \right) \leq \varepsilon(\Delta t) \right) = 1
$$
holds for all $t\in[0,T-\Delta t]$, all $s\in[t,t+\Delta t]$, all stopping times $t\leq\tau\leq s$, and all $\alpha\in\mathbb{A}$.
\end{ass}

The above assumption is satisfied, for example, if $\sigma$ is an It\^{o} process with bounded drift and diffusion coefficients.
Next, we state a continuity assumption on the drift, which can be interpreted as a uniform right-continuity in probability of the martingale $\tilde{\EE}^{\alpha}_{t} \mu^{\alpha}_s$.

\begin{ass}\label{ass:main.mu.cont.strong}
For any $\alpha\in\mathbb{A}$ and any $t\in[0,T)$, there exists a deterministic function $\varepsilon(\cdot)\geq0$, such that $\varepsilon(\Delta t)\rightarrow0$, as $\Delta t\rightarrow0$, and, $\PP^{\alpha}$-a.s.,
$$
\tilde{\PP}^{\alpha}_{t'} \left( \left| \int_{t}^T\left(\tilde{\EE}^{\alpha}_{t''}  \mu^{\alpha}_s - \tilde{\EE}^{\alpha}_{t'} \mu^{\alpha}_s\right) ds\right| \geq \varepsilon(\Delta t)\right) \leq \varepsilon(\Delta t)
$$
holds for all $t\leq t' \leq t'' \leq t+\Delta t\leq T$.
\end{ass}

Notice that Assumptions \ref{ass:joint.cond.reg}, \ref{ass:main.L2.strong}, and \ref{ass:main.mu.cont.strong} are not quite standard. Therefore, below, we describe a more specific (although, still, rather general) diffusion-based framework, in which the Assumptions \ref{ass:sigma}--\ref{ass:main.mu.cont.strong} reduce to standard regularity conditions on the diffusion coefficients, and are easily verified.
To this end, consider a model in which $\mu^{\alpha}_t = \bar{\mu}^{\alpha}(t,Y_t)$, $\sigma_t = \bar{\sigma}(t,Y_t)$, and, under $\PP$, the process $Y$ is a diffusion taking values in $\RR^d$
$$
dY_t = \Gamma(t,Y_t)dt + \Sigma(t,Y_t) d\bar{B}_t,
$$
where $\Gamma:[0,T]\times\RR^d\rightarrow\RR^d$, $\Sigma=(\Sigma^{i,j})$ is a mapping on $[0,T]\times\RR^d$ with values in the space of $d\times m$ matrices, and $\bar{B}$ is $m$-dimensional Brownian motion under $\PP$ (on the original stochastic basis). We assume that $\Gamma$ and $\Sigma$ possess enough regularity to conclude that $Y$ is a strongly Markov process.
Notice that Assumptions \ref{ass:sigma} and \ref{ass:A.alpha} reduce to the upper and lower bounds on the functions $\bar{\mu}^{\alpha}$ and $\bar{\sigma}$.
Assumption \ref{ass:joint.cond.reg} is satisfied if we assume that $\PP^{\alpha}\sim\PP$, for all $\alpha\in\mathbb{A}$.
Let us further assume that the Radon-Nikodym derivative of each measure is in Girsanov form:
$$
\frac{d\PP^{\alpha}}{d\PP} = \exp\left(-\frac{1}{2} \int_0^t \|\gamma^{\alpha}(s,Y_s)\|^2 ds + \int_0^t \gamma^{\alpha}(s,Y_s) d\bar{B}_s \right),
$$
with an $\RR^d$-valued function $\gamma^{\alpha}$, for each $\alpha\in\mathbb{A}$.
Let us assume that all entries of $\Gamma$, $\gamma^{\alpha}$ and $\Sigma$ are absolutely bounded by a constant (uniformly over $\alpha\in\mathbb{A}$).
Assuming, in addition, that $\bar{\sigma}$ is globally Lipschitz, we easily verify Assumption \ref{ass:main.L2.strong}.
In order to verify Assumption \ref{ass:main.mu.cont.strong}, we assume that the quadratic form generated by $A(t,y):=\Sigma(t,y) \Sigma^T(t,y)$ is bounded away from zero, uniformly over all $(t,y)$, and that the entries of $\Gamma$, $\gamma^{\alpha}$ and $\Sigma$ are continuously differentiable with absolutely bounded derivatives (uniformly over $\alpha\in\mathbb{A}$).
Then, the Feynman-Kac formula implies that, for any $t\leq s$,
$$
\tilde{\EE}^{\alpha}_t \mu^{\alpha}_s = u^{s,\alpha}(t,Y_t),
$$
where $u^{s,\alpha}$ is the unique solution to the associated partial differential equation (PDE)
$$
\partial_t u^{s,\alpha} + \sum_{i=1}^d \Gamma^{\alpha,i} \partial_{y_i} u^{s,\alpha} + \frac{1}{2}\sum_{i,j=1}^d A^{i,j} \partial^2_{y_i y_j} u^{s,\alpha} = 0,\,\,\,\,(t,y)\in (0,s)\times\RR^d,
\quad u^{s,\alpha}(s,y)=\bar{\mu}^{\alpha}(s,y),
$$
and $\Gamma^{\alpha}=\Gamma + \Sigma \gamma^{\alpha}$.
Assume that, for each $s\in[0,T]$, the function $\bar{\mu}^{\alpha}(s,\cdot)$ is continuously differentiable with absolutely bounded derivatives, uniformly over all $(s,\alpha)$. Then, the standard Gaussian estimates for derivatives of the fundamental solution to the above PDE (cf. Theorem 9.4.2 in \cite{Friedman.book}) imply that every $\partial_{y_i}u^{s,\alpha}$ is absolutely bounded, uniformly over all $(s,\alpha)$.
Then, It\^o's formula and It\^o's isometry yield, for all $t'\leq t''$ and $s\geq t'$:
$$
\tilde{\EE}^{\alpha}_{t'}\left(\tilde{\EE}^{\alpha}_{t''}  \mu^{\alpha}_s - \tilde{\EE}^{\alpha}_{t'} \mu^{\alpha}_s\right)^2
= \sum_{j=1}^m \int_{t'}^{t''\wedge s} \tilde{\EE}^{\alpha}_{t'}\left(\sum_{i=1}^d\partial_{y_i} u^{s,\alpha}(v,Y_v)\Sigma^{i,j}(v,Y_v) \right)^2 dv \leq C_1 (t''\wedge s\,-\,t'),
$$
with some constant $C_1>0$.
%It is easy to see that the right hand side of the above vanishes as $t''-t'\rightarrow0$, uniformly over all $0\leq t'\leq s\leq T$.
The above estimate and Jensen's inequality imply the statement of Assumption \ref{ass:main.mu.cont.strong} and complete the description of the diffusion-based setting.
%The above assumption is satisfied, for example, if $\tilde{\EE}^{\alpha}_{t} \mu^{\alpha}_s$ has an integral representation with respect to some $\PP^{\alpha}$-Brownian motion $B^{\alpha}$,
%$$
%\tilde{\EE}^{\alpha}_{t} \mu^{\alpha}_s = \mu^{\alpha}_s + \int_t^s \beta^{\alpha,s}_u dB^{\alpha}_u,
%$$
%and $\beta^{\alpha,s}_u$ is absolutely bounded by a constant, for all $(u,s)$. The latter condition is satisfied, for example, in any diffusion-based model -- i.e. where $\mu^{\alpha}_t=\mu^{\alpha}(t,Y_t)$, with a diffusion process $Y$ -- in which the generator of $Y$ is strictly elliptic, its coefficients are H\"older continuous, and the derivative of $\mu^{\alpha}(t,\cdot)$ is absolutely bounded, uniformly over $t$. This observation follows from the Feynman-Kac formula and the standard Schauder estimates.

As in Section \ref{se:examples}, we also consider a progressively measurable random field $\tilde{D}$, s.t. $\PP$-a.s. the function $\tilde{D}_t(\cdot)-\tilde{D}_s(\cdot)$ is strictly decreasing and vanishing at zero, for any $0\leq s < t \leq T$. We assume that the demand curve, $\tilde{D}_t(\cdot)-\tilde{D}_s(\cdot)$, cannot be ``too flat".

\begin{ass}\label{ass:main.demandInv.unif}
There exists $\varepsilon>0$, s.t., for any $0\leq t - \varepsilon \leq s < t \leq T$, there exists a $\tilde{\mathcal{F}}_{s}\otimes \mathcal{B}(\RR)$-measurable random function $\kappa_s(\cdot)$, s.t., $\PP$-a.s., $\kappa_{s}(\cdot)$ is strictly decreasing and $\left|\tilde{D}_t(p)-\tilde{D}_s(p)\right| \geq \left| \kappa_{s}(p) \right|$, for all $p\in\RR$.
\end{ass}

Finally, we introduce the empirical distribution process $(\tilde{\mu}_t)$, with values in the space of finite sigma-additive measures on $\mathbb{S}$.
The next assumption states that every $\tilde{\mu}_t$ is dominated by a deterministic measure. 

\begin{ass}\label{ass:dom.mu}
For any $t\in[0,T]$, there exists a finite sigma-additive measure $\mu^{0}_t$ on $\left(\mathbb{S}, \mathcal{B}\left(\mathbb{S} \right)\right)$, s.t., $\PP$-a.s., $\tilde{\mu}_t$ is absolutely continuous w.r.t. $\mu^0_t$.
\end{ass}

We partition the time interval $[0,T]$ into $N$ subintervals of size $\Delta t=T/N$.
A discrete time model is obtained by discretizing the continuous time one 
$$
\mathcal{F}_n = \tilde{\mathcal{F}}_{n\Delta t},\quad p^0_n = \tilde{p}^0_{n\Delta t},\quad D_n(p) = (\tilde{D}_{n\Delta t}-\tilde{D}_{(n-1)\Delta t})(p-p^0_n),\quad \mu_n = \tilde{\mu}_{n\Delta t}.
%, \quad \PP^{\alpha}_n = \tilde{\PP}^{\alpha}_{n\Delta t},\quad \EE^{\alpha}_n
$$
Before we present the main results, let us comment on the above assumptions. These assumptions are important from a technical point of view, however, some of them have economic interpretation that may provide (partial) intuitive explanations of the results that follow. In particular, Assumption \ref{ass:sigma} ensures that the fundamental price remains ``noisy," which implies that an agent can execute a limit order very quickly by posting it close to the present value of $p^0$, if there are no other orders posted there. In combination with Assumption \ref{ass:main.demandInv.unif}, the latter implies that, when the frequency, $N$, is high, an agent has a lot of opportunities to execute her limit order at a price close to the fundamental price (at least, if no other orders are posted too close to the fundamental price). Intuitively, this means that the agent's execution value should improve as the frequency increases.
Assumption \ref{ass:main.mu.cont.strong} ensures that, if an agent has a signal about the direction of the fundamental price, this signal is persistent -- i.e., it is continuous in the appropriate sense. When the trading frequency $N$ is large, such persistency means that an agent has a large number of opportunities to exploit the signal, implying that she is in no rush to have her order executed immediately. The main results of this work, presented below, along with their proofs, confirm that these heuristic conclusions are, indeed, correct. 

As mentioned in the preceding sections, our main goal is to analyze the liquidity effects of increasing the trading frequency. Therefore, we fix a limiting continuous time model, and consider a sequence of discrete time models, obtained from the limiting one as described above, for $N\rightarrow\infty$. This can be interpreted as observing the same population of agents, each of whom has a fixed continuous time model for future demand, in various exchanges that allow for different trading frequencies.
We begin with the following theorem, which shows that, if every market model in a given sequence admits a non-degenerate equilibrium, then, the terminal bid and ask prices converge to the fundamental price, as the trading frequency goes to infinity.

\begin{theorem}\label{le:main.zeroTermSpread}
Let Assumptions \ref{ass:sigma}, \ref{ass:A.alpha}, \ref{ass:joint.cond.reg}, \ref{ass:main.L2.strong}, \ref{ass:main.demandInv.unif}, \ref{ass:dom.mu} hold. 
Consider a family of uniform partitions of a given time interval $[0,T]$, with diameters $\left\{\Delta t=T/N>0\right\}$ and with the associated family of discrete time models, and denote the associated fundamental price process by $p^{0,\Delta t}$. 
Assume that every such model admits a non-degenerate LTC equilibrium, and denote the associated bid and ask prices by $p^{b,\Delta t}$ and $p^{a,\Delta t}$ respectively.
Then, there exists a deterministic function $\varepsilon(\cdot)$, s.t. $\varepsilon(\Delta t)\rightarrow0$, as $\Delta t\rightarrow0$, and, for all small enough $\Delta t>0$, the following holds $\PP$-a.s.:
$$
\left|p^{a,\Delta t}_{N} - p^{0,\Delta t}_{N}\right| + \left|p^{b,\Delta t}_{N} - p^{0,\Delta t}_{N}\right| \leq \varepsilon(\Delta t)
$$
\end{theorem}

The above theorem has a useful corollary, which can be interpreted as follows: \emph{if the market does not degenerate as the frequency increases, then, such an increase improves market efficiency}. Here, we understand the ``improving efficiency" in the sense that the expected execution price (i.e., the price per share that an agent expects to receive or pay by the end of the game) of every agent converges to the fundamental price.

\begin{cor}\label{prop:main.smallspread}
Under the assumptions of Theorem \ref{le:main.zeroTermSpread}, denote the support of every equilibrium by $\tilde{\mathbb{A}}^{\Delta t}$ and the associated expected execution prices by $\lambda^{a,\Delta t}$ and $\lambda^{b,\Delta t}$. Then, there exists a deterministic function $\varepsilon(\cdot)$, such that $\varepsilon(\Delta t)\rightarrow0$, as $\Delta t\rightarrow0$, and, $\PP$-a.s.,
$$
\sup_{n=0,\ldots,N,\,\alpha\in\tilde{\mathbb{A}}^{\Delta t}}\left(\left|\lambda^{a,\Delta t}_n(\alpha) - p^{0,\Delta t}_n\right| + \left|\lambda^{b,\Delta t}_n(\alpha) - p^{0,\Delta t}_n\right|\right) \leq \varepsilon(\Delta t),
$$
for all small enough $\Delta t>0$.
\end{cor}
\begin{proof}
Denote $\EE^{\alpha}_n = \tilde{\EE}^{\alpha}_{n\Delta t}$.
It follows from Corollary \ref{cor:piecewiseLin}, in Appendix A, and the definition of LTC equilibrium that $\lambda^{a,\Delta t}_{N}(\alpha) = p^{b,\Delta t}_N$ and $\lambda^{b,\Delta t}_{N}(\alpha)=p^{a,\Delta t}_N$. It also follows from Corollary \ref{cor:piecewiseLin} (or, more generally, from the definition of a value function) that $\lambda^{a,\Delta t}(\alpha)$ is a supermartingale, and $\lambda^{b,\Delta t}(\alpha)$ is a submartingale, under $\PP^{\alpha}$. Thus, we have: $\lambda^{a,\Delta t}_{n}(\alpha) \geq \EE^{\alpha}_{n} p^{b,\Delta t}_{N}$ and $\lambda^{b,\Delta t}_{n}(\alpha) \leq \EE^{\alpha}_{n} p^{a,\Delta t}_{N}$.
On the other hand, notice that we must have: $\lambda^{a,\Delta t}_{n}(\alpha) \leq \EE^{\alpha}_{n} p^{a,\Delta t}_{N}$ and $\lambda^{b,\Delta t}_{n}(\alpha) \geq \EE^{\alpha}_{n} p^{b,\Delta t}_{N}$. Assume, for example, that $\lambda^{a,\Delta t}_{n}(\alpha) > \EE^{\alpha}_{n} p^{a,\Delta t}_{N}$ on the event $\Omega'$ of positive $\PP^{\alpha}$-probability. 
Consider an agent at state $(0,\alpha)$, who follows the optimal strategy of an agent at state $(1,\alpha)$, starting from time $n$ and onward, on the event $\Omega'$ (otherwise, she does not do anything). It is easy to see that the objective value of this strategy is 
$$
\EE^{\alpha}\left( \bone_{\Omega'} \left( \lambda^{a,\Delta t}_{n}(\alpha) - \EE^{\alpha}_{n} p^{a,\Delta t}_{N} \right)\right) > 0,
$$  
which contradicts Corollary \ref{cor:piecewiseLin}. The second inequality is shown similarly. Thus, we conclude that, for any $n=0,\ldots,N-1$, both $\lambda^{a,\Delta}_{n}(\alpha)$ and $\lambda^{b,\Delta}_{n}(\alpha)$ belong to the interval 
$$
\left[\EE^{\alpha}_{n} p^{b,\Delta t}_{N},\,\EE^{\alpha}_{n} p^{a,\Delta t}_{N}\right],
$$
which, in turn, converges to zero, as $\Delta t\rightarrow0$, due to the deterministic bounds obtained in the proof of Proposition \ref{le:main.zeroTermSpread}.
\qed
\end{proof}

The results of Theorem \ref{le:main.zeroTermSpread} and Corollary \ref{prop:main.smallspread} can be viewed as a specific case of a more general observation: markets become more efficient as the frictions become smaller. In the present setting, the limited trading frequency is viewed as friction, and the market efficiency is measured by the difference between the bid and ask prices, or between the expected execution prices. Many more instances of analogous results can be found in the literature, depending on the choice of a friction type. For example, the markets become efficient in \cite{MMS.gmm1} and \cite{MMS.gmm2} as the number of insiders vanishes. Similarly, the markets become efficient in \cite{DA.DuZhu} as the trading frequency increases and the size of private signals vanishes. It is also mentioned in \cite{MMS.gliq1} that the market would become efficient if there was no restriction on the size of agents' inventories therein.

The above results demonstrate the positive role of high trading frequency. However, they are based on the assumption that the market does not degenerate as frequency increases. In the context of Section \ref{se:examples}, we saw that the markets do not degenerate only if the agents are market-neutral (i.e. $\alpha=0$). If this condition is violated and the frequency $N$ is sufficiently high, the market does not admit any non-degenerate equilibrium (i.e., there exists no safe regime, in which the liquidity crisis would never occur). It turns out that this conclusion still holds in the general setting considered herein.

\begin{theorem}\label{thm:main.necessary}
Let Assumptions \ref{ass:sigma}, \ref{ass:A.alpha}, \ref{ass:joint.cond.reg}, \ref{ass:main.L2.strong}, \ref{ass:main.mu.cont.strong}, \ref{ass:main.demandInv.unif}, \ref{ass:dom.mu} hold.
Consider a family of uniform partitions of a given time interval $[0,T]$, with diameters $\left\{\Delta t=T/N>0\right\}$, containing arbitrarily small $\Delta t$, and with the associated family of discrete time models. 
Assume that every such model admits a non-degenerate LTC equilibrium, with the same support $\tilde{\mathbb{A}}$. Then, for all $\alpha\in\tilde{\mathbb{A}}$, we have: $\tilde{p}^{0}$ is a {\bf martingale} under $\PP^{\alpha}$.
\end{theorem}

%The proof of the above theorem has an interesting corollary.

%\begin{cor}\label{cor:fundprice.midpoint}
%Under the assumptions of Theorem \ref{thm:main.necessary}, for any $n=0,\ldots,N$ and any $\alpha\in\tilde{\mathbb{A}}$, the following holds $\PP^{\alpha}$-a.s.:
%$$
%p^b_n < p^0_n < p^a_n
%$$
%\end{cor}
%Notice that in many existing studies of market microstructure, the fundamental price process is \emph{defined} to be a \emph{mid-point} (or just a point) between the bid and ask prices. Herein, we obtain this conclusion as an \emph{output of the equilibrium}.

The above theorem shows that the market degenerates even if the signal $\mu^{\alpha}$ is very small (but non-zero), provided the trading frequency $N$ is large enough.
Therefore, as discussed at the end of Section \ref{se:examples}, such degeneracy cannot be attributed to any fundamental reasons, and we refer to it as the \emph{endogenous liquidity crisis}.
Let us provide an intuitive (heuristic) argument for why the statement of Theorem \ref{thm:main.necessary} holds. %Consider any family of non-degenerate equilibria with arbitrarily high trading frequencies. 
Assume, first, that all long agents (i.e., those having positive inventory) are bullish about the asset (i.e., have a positive drift $\mu^{\alpha}$).
%Then, as the market efficiency increases, their expected execution prices
%The proof of Theorem \ref{thm:main.necessary} also reveals what causes the degeneracy. 
Then, similar to Section \ref{se:examples}, the higher trading frequency amplifies the \emph{adverse selection effect}, forcing the long agents to withdraw liquidity from the market (i.e., they prefer to do nothing and wait for a higher fundamental price level).
Note that, in the present setting, the agents may have different beliefs, the LOB may have a complicated shape and dynamics, and the expected execution prices are no longer deterministic. All this makes it difficult to provide a simple description of how the high frequency amplifies the adverse selection. Nevertheless, the general analysis of this case is still based on the idea discussed at the end of Section \ref{se:examples}: it has to do with how fast $\tilde{\EE}^{\alpha}_{n\Delta t} (p^0_{n+1} - p\, |\,p^0_{n+1}>p )$ vanishes (as the frequency increases), relative to the rate at which the expected execution prices approach the fundamental price.
Thus, to avoid market degeneracy, there must be a non-zero mass of long agents who are market-neutral or bearish.
As the trading frequency grows, these agents will post their limit orders at lower levels.
Next, assume that there exists a bullish agent (long, short, or with zero inventory). Then, at a sufficiently high trading frequency, the agent's expected value of a long position in a single share of the asset will exceed the ask prices posted by the market-neutral and bearish long agents.
%Then, at a sufficiently high trading frequency, the expected execution price of this agent will exceed the posted ask prices of the neutral long agents. 
In this case, the bullish agent prefers to buy more shares at the posted ask price, in order to sell them later. As the agents are small and their objectives are linear, the bullish agent can scale up her strategy to generate infinite expected profits. This contradicts the definition of optimality and implies that an equilibrium fails to exist.
Thus, all agents have to be either market-neutral or bearish. Applying a symmetric argument, we conclude that all agents must be market-neutral.\footnote{This argument, along with the fact that Definition \ref{def:optControl} requires an optimal control to be optimal for \emph{all} $\alpha$, explains why the statement of Theorem \ref{thm:main.necessary} holds for \emph{all}, as opposed to $\mu_n$-a.e., $\alpha\in\tilde{\mathbb{A}}$.}
A rigorous formulation of the above arguments, which constitutes the proof of Theorem \ref{thm:main.necessary}, is given in Section \ref{se:pf.2}.
%A similar argument applies when all short agents (i.e. those having negative inventory) are bearish about the asset (i.e. have a negative drift $\mu^{\alpha}$).
%It is also clear intuitively that, if all long agents are bearish and all short agents are bullish or market-neutral, then, at a sufficiently high trading frequency, the bid prices posted by the short agents will become attractive to the long agents. As a result, the long agents will prefer to submit market orders, and there will be no limit orders left, at least, on one side of the book. Thus, there must be a non-zero mass of long agents who are market-neutral. Similarly, we conclude that there must be a non-zero mass of short agents who are market-neutral.
%A similar argument applies if there exists a bearish agent. Thus, all agents have to be market-neutral. 

It is worth mentioning that the possible degeneracy of the LOB is also documented in \cite{MMS.gmm1}, and is referred to as a ``market shut down". The setting used in the latter paper is very different: it analyzes a quote-driven exchange (i.e., the one with a designated market maker) and assumes the existence of insiders with superior information. Nevertheless, it is possible to draw a parallel with the LOB degeneracy in the present setting. Namely, the degeneracy in \cite{MMS.gmm1} occurs when the number of insiders increases, which implies that the signal, generated by the insiders' trading, becomes sufficiently large. The latter is similar to the deviation from martingality of the fundamental price in the present setting. However, an increase in the number of insiders in \cite{MMS.gmm1} also implies an increase in frictions (since the insiders can be interpreted as friction in \cite{MMS.gmm1}). Theorem \ref{thm:main.necessary}, on the other hand, states that a market degeneracy will occur when the frictions are sufficiently small. Perhaps, this dual role of the number of insiders did not allow for a detailed analysis of market shut downs in \cite{MMS.gmm1}. Many other models of market microstructure (cf. \cite{MMS.g3}, \cite{MMS.g6}, \cite{MMS.g1}, \cite{MMS.g2}, \cite{DA.DuZhu}) are not well suited for the analysis of market degeneracy, either because the agents in these models pursue ``one-shot" strategies (i.e., they cannot choose to wait and post a limit order later) or because the fundamental price (or its analogue) is restricted to be a martingale.

\section{Conditional tails of the marginal distributions of It\^{o} processes}
\label{se:tails}

As follows from the discussion in the preceding sections, in order to prove the main results of the paper, we need to investigate the properties of marginal distributions of the fundamental price $\tilde{p}^0$ (more precisely, the distributions of its increments).
In order to prove Theorem \ref{le:main.zeroTermSpread}, we need to show that the difference between the fundamental price and the bid or ask prices converges to zero, as the frequency $N$ increases to infinity. It turns out that, for this purpose, it suffices to show that the distribution of a normalized increment of $\tilde{p}^0$ converges to the standard normal distribution.
The following lemma summarizes these results. It is rather simple, but technical, hence, its proof is postponed to Appendix B.
%We begin with several useful technical lemmas. In words, they show that, in the present setting, the conditional distributions of the normalized increments of the fundamental price are asymptotically Gaussian.
In order to formulate the result (and to facilitate the derivations in subsequent sections), we introduce addiitonal notation.
For notational convenience, we drop the superscript $\Delta t$ for some variables (we only emphasize this dependence when it is important).
For any market model on the time interval $[0,T]$, associated with a uniform partition with diameter $\Delta t=T/N>0$, and having a fundamental price process $p^0$, we define
\begin{equation}\label{eq.xi.not}
\xi_n = p^{0}_n - p^{0}_{n-1} = \tilde{p}^0_{t_n} - \tilde{p}^0_{t_{n-1}},
%\quad \tilde{\xi}_n = \tilde{p}^0_{t_n} - \tilde{p}^0_{t_{n-1}},
\quad \EE^{\alpha}_n = \tilde{\EE}^{\alpha}_{t_n},
\quad \PP^{\alpha}_n = \tilde{\PP}^{\alpha}_{t_n},
\quad t_n = n \Delta t,
\quad n=1,\ldots,NT/\Delta t.
\end{equation}
%The next lemma shows that the conditional distribution of a fundamental price increment becomes close to normal, as the trading frequency increases. In order to formulate it, 
We denote by $\eta_0$ a standard normal random variable (on a, possibly, extended probability space), which is independent of $\mathcal{F}_N$ under every $\PP^{\alpha}$.

\begin{lemma}\label{gapproxapplied}
Let Assumptions \ref{ass:sigma}, \ref{ass:A.alpha}, \ref{ass:joint.cond.reg}, \ref{ass:main.L2.strong} hold.
Then, there exists a function $\varepsilon(\cdot)\ge0$, s.t. $\varepsilon(\Delta t)\to0$, as $\Delta t\to0$, and the following holds $\PP$-a.s., for all $p\in\RR$, all $\alpha\in\mathbb{A}$, and all $n=1,\ldots,N$,
\begin{itemize}

\item[(i)] $(|p|\vee 1)\left|\PP^{\alpha}_{n-1}\left(\frac{\xi_n}{\sqrt{\Delta t}} >p\right) 
- \PP^{\alpha}_{n-1}\left(\sigma_{t_{n-1}}\eta_0>p\right)\right|
\le\varepsilon(\Delta t)$,

\item[(ii)] $\left|\EE^{\alpha}_{n-1}\left( \frac{\xi_n}{\sqrt{\Delta t}}\bone_{\left\{\xi_n/\sqrt{\Delta t}>p\right\}}\right) 
- \EE^{\alpha}_{n-1}\left(\sigma_{t_{n-1}}\eta_0 \bone_{\left\{\sigma_{t_{n-1}}\eta_0>p\right\}}\right) \right|\le\varepsilon(\Delta t)$.
%\item[(iii)] $\left| p \PP^{\alpha}_{n-1} \left(\frac{\txi_n}{\sqrt{\Delta t}}>p\right) 
%- p \PP^{\alpha}_{n-1}\left(\sigma_{t_{n-1}}\eta_0>p\right)\right|
%\le\varepsilon(\Delta t)$
\end{itemize}
In addition, the above estimates hold if we replace $(\xi_n,\eta_0,p)$ by $(-\xi_n,-\eta_0,-p)$.
\end{lemma}
%\begin{proof}
%The proof is given in Appendix A.
%\qed
%\end{proof}

In order to prove Theorem \ref{thm:main.necessary} we need to compare the rates at which the conditional expectations $\EE^{\alpha}_n (p^0_{n+1} - p\, |\,p^0_{n+1}>p )$ vanish (as the frequency $N$ goes to infinity) to the rate at which the expected execution prices converge to the fundamental price. This requires a more delicate analysis -- in particular, the mere proximity of the distribution of a (normalized) fundamental price increment to the Gaussian distribution is no longer sufficient. In fact, what we need is a precise uniform estimate of the conditional tail of the distribution of a fundamental price increment. The desired property is formulated in the following lemma, which, we believe, is valuable in its own right. 
This result enables us to estimate the tails of the conditional marginal distribution of an It{\^o} process $X$ uniformly by an exponential. To the best of our knowledge, this result is new. The main difficulties in establishing the desired estimates are: (a) the fact that we estimate the \emph{conditional}, as opposed to the regular, tail, and (b) the fact that the estimates need to be uniform over the values of the argument. Note that, even in the case of a diffusion process $X$, the classical Gaussian-type bounds for the tails of the marginal distributions of $X$ are not sufficient to establish the desired estimates. The reason is that, in general, the Gaussian estimates of the regular tails from above and from below have different orders of decay, for the large values of the argument, which makes them useless for estimating the conditional tail (which is a ratio of two regular tails).

%In short, this lemma shows that, if the quadratic variation of a continuous semimartingale does not oscillate too much, and if its drift is small enough, then, the marginal distributions of this process and its running maximum behave similar to those of a Brownian motion. This conclusion sounds very natural, to the point that it may even seem trivial. However, the precise formulation of the desired similarity makes the result far from obvious. Namely, the first statement of the lemma is meant to provide \emph{uniform} estimates of the tails of the running maximum of a process $X$, conditional on the event that this maximum exceeds a given threshold. The second statement connects the tails of a marginal distribution of this process with the tails of its running maximum -- it can be viewed as an asymptotic version of the Bachelier theorem for a continuous semimartingale. This connection, then, allows us to estimate the tails of the conditional marginal distribution of $X$, in the same way it is done for the running maximum of $X$ (cf. Lemma \ref{le:necessary.2}). Note that, even in the case of a diffusion process $X$, the classical Gaussian-type bounds for the tails of the marginal distributions of $X$ are not sufficient to establish the desired estimates, as the bounds from above and from below have different orders of decay, for the large values of the argument. The main difficulty in establishing the desired estimates is their uniformity over the values of the argument. 

\begin{lemma}\label{le:necessary.marginal.maximum}
Consider the following continuous semimartingale on a stochastic basis $(\hat{\Omega},(\hat{\mathcal{F}}_t)_{t\in[0,1]},\hat{\PP})$:
$$
X_t = \int_0^t \hat{\mu}_u du + \int_{0}^t \hat{\sigma}_u dB_u,\,\,\,\,\,\,\,\,\,\,\,\,t\in[0,1],
$$
where $B$ is a Brownian motion (with respect to the given stochastic basis), $\hat{\mu}$ and $\hat{\sigma}$ are progressively measurable processes, such that the above integrals are well defined.
%, with deterministic $\tilde{\sigma}_0>0$. 
%Assume that, for any stopping time $\tau$, with values in $[0,1]$, the measure $\tilde{\PP}$ admits a regular conditional probability with respect to $\tilde{\mathcal{F}}_{\tau}$, denoted $\tilde{\PP}_{\tau}$.
Assume that, for any stopping time $\tau$ with values in $[0,1]$, $c\leq |\hat{\sigma}_{\tau}| \leq C$ holds a.s. with some constants $c,C>0$. 
Then, there exists $\varepsilon>0$, depending only on $(c,C)$, s.t., if
$$
\hat{\mu}^2_{\tau}\leq \varepsilon,\quad \hat{\EE}\left( (\hat{\sigma}_{s\vee\tau} - \hat{\sigma}_{\tau})^2 \,|\,\hat{\mathcal{F}}_{\tau} \right) \leq \varepsilon\,\,\,\, \text{a.s.},
$$
for all $s\in[0,1]$ and all stopping time $\tau$, with values in $[0,1]$, then, for any $c_1>0$, there exists $C_1>0$, depending only on $(c,C,\varepsilon,c_1)$, s.t. the following holds:
$$
\hat{\PP}(X_1 > x+z\,\vert\, X_1>x) \leq C_1 e^{-c_1 z},\quad\forall x,z\geq0.
$$
\end{lemma}
\begin{proof}
In the course of this proof, we will use the shorthand notation, $\hat{\EE}_{\tau}$ and $\hat{\PP}_{\tau}$, to denote the conditional expectation and the conditional probability w.r.t $\hat{\mathcal{F}}_{\tau}$.
We also denote
$$
A_t = \int_0^t \hat{\mu}_u du,
\,\,\,\,\,\,\,\,\,\,\,\,\,\,\,\,G_t = \int_{0}^t \hat{\sigma}_u dB_u.
$$
For any $x\geq0$, let us introduce $\tau_x = 1\wedge\inf\left\{t\in[0,1]\,:\, X_t = x \right\}$.
Then
$$
\hat{\PP}(X_1 > x+z) 
\leq \hat{\PP}(\sup_{t\in[0,1]} X_t > x+z)
= \hat{\EE} \left( \bone_{\left\{ \tau_x < 1 \right\}} \hat{\PP}_{\tau_x} \left(\sup_{s\in[\tau_x,1]} (X_s - x) > z \right) \right)
$$
Notice that, on $\left\{\tau_x\leq s \right\}$, we have: $X_s - x = A_{s\vee\tau_x} - A_{\tau_x} + G_{s\vee\tau_x} - G_{\tau_x}$.
In addition, the process $(Y)_{s\in[0,1]}$, with $Y_s = A_{s\vee\tau_x} - A_{\tau_x}$, is adapted to the filtration $(\hat{\mathcal{F}}_{\tau_x\vee s})$, while the process $(Z)_{s\in[0,1]}$, with $Z_s = G_{s\vee\tau_x} - G_{\tau_x}$, is a martingale with respect to it.
Next, on $\left\{\tau_x < 1 \right\}$, we have:
$$
\hat{\PP}_{\tau_x} \left(\sup_{s\in[\tau_x,1]} (X_s - x) > z \right)
= \hat{\PP}_{\tau_x} \left(\sup_{s\in[0,1]} (Y_s + Z_s) > z\right)
$$
$$
\leq \hat{\PP}_{\tau_x} \left(\sup_{s\in[0,1]} \exp\left(c_1Z_s - \frac{1}{2}c_1^2\langle Z\rangle_s \right) > \exp\left( c_1z - c_1\sqrt{\varepsilon} - \frac{1}{2} c_1^2C^2\right)\right),
$$
where we used the fact that $\langle Z\rangle_s \leq \langle X\rangle_1 \leq C^2$, for all $s\in[0,1]$. Using the Novikov's condition, it is easy to check that
$$
M_s = \exp\left(c_1Z_s - \frac{1}{2}c_1^2\langle Z\rangle_s \right),\,\,\,\,\,\,\,\,\,s\in[0,1],
$$
is a true martingale, and, hence, we can apply the Doob's martingale inequality to obtain, on $\left\{\tau_x < 1 \right\}$:
$$
\hat{\PP}_{\tau_x} \left(\sup_{s\in[0,1]} \exp\left(c_1Z_s - \frac{1}{2}c_1^2\langle Z\rangle_s \right) > \exp\left( c_1z - c_1\sqrt{\varepsilon} - \frac{1}{2}c_1^2 C^2\right)\right)
\leq \exp\left( -c_1z + c_1\sqrt{\varepsilon} + \frac{1}{2} c_1^2 C^2\right).
$$
Collecting the above inequalities, we obtain
\begin{equation}\label{eq.necessary.biglemma.step1.res}
\hat{\PP}(X_1 > x+z) \leq \hat{\PP}(\sup_{t\in[0,1]} X_t > x+z) \leq C_2(\varepsilon) e^{-c_1z} \hat{\PP}(\tau_x < 1) = C_2(\varepsilon) e^{-c_1z} \hat{\PP}(\sup_{t\in[0,1]} X_t > x).
\end{equation}
The next step is to estimate the distribution tails of a running maximum via the tails of the distribution of $X_1$.
To do this, we proceed as before:
\begin{equation}\label{eq.necessary.biglemma.step2.1}
\hat{\PP}(X_1 > x) 
= \hat{\EE} \left( \bone_{\left\{ \tau_x < 1 \right\}} \hat{\PP}_{\tau_x}\left(Y_1 + Z_1 > 0\right) \right),
\end{equation}
with $Y$ and $Z$ defined above.
Notice that, on $\left\{\tau_x < 1 \right\}$,
$$
\hat{\PP}_{\tau_x}\left(Y_1 + Z_1 > 0\right) 
= \hat{\PP}_{\tau_x}\left(\hat{\sigma}_{\tau_x} \frac{B_1-B_{\tau_x}}{\sqrt{1-\tau_x}} 
+ \frac{1}{\sqrt{1-\tau_x}} \int_{\tau_x}^{1} \hat{\mu}_{u} du
+ \frac{1}{\sqrt{1-\tau_x}} \int_{0}^{1} (\hat{\sigma}_{u\vee\tau_x} - \hat{\sigma}_{\tau_x}) dB^x_u > 0\right),
$$
where $B^x_s = B_{s\vee\tau_x}$ is a continuous square-integrable martingale with respect to $(\hat{\mathcal{F}}_{s\vee\tau_x})$.
Denote
$$
R_s = \int_{0}^{s} (\hat{\sigma}_{u\vee\tau_x} - \hat{\sigma}_{\tau_x}) dB^x_u,
\quad s\in[0,1],
$$
and notice that it is a square-integrable martingale with respect to $(\hat{\mathcal{F}}_{s\vee\tau_x})$. 
Then, on $\left\{\tau_x<1\right\}$ (possibly, without a set of measure zero), we have:
$$
\hat{\EE}_{\tau_x} \left(\frac{1}{\sqrt{1-\tau_x}} R_1 \right)^2
=\frac{1}{1-\tau_x} \hat{\EE}_{\tau_x} R^2_1
\leq \frac{1}{1-\tau_x}\int_{\tau_x}^{1} \hat{\EE}_{\tau_x}(\hat{\sigma}_{u\vee\tau_x} - \hat{\sigma}_{\tau_x})^2 du
\leq \varepsilon.
$$
In addition,
$$
\hat{\EE}_{\tau_x} \left( \frac{1}{\sqrt{1-\tau_x}} \int_{\tau_x}^{1} \hat{\mu}_{u} du \right)^2
\leq \varepsilon.
$$
Collecting the above and using Chebyshev's inequality, we obtain, on $\left\{\tau_x < 1 \right\}$:
$$
\left|\hat{\PP}_{\tau_x}\left(Y_1 + Z_1 > 0\right) 
- \hat{\PP}_{\tau_x}\left(\hat{\sigma}_{\tau_x} \frac{B_1-B_{\tau_x}}{\sqrt{1-\tau_x}} \leq -\varepsilon^{1/3} \right)\right|
\leq 2\varepsilon^{1/6}.
$$
On the other hand, due to the strong Markov property of Brownian motion, on $\left\{\tau_x<1\right\}$, we have, a.s.:
$$
\hat{\PP}_{\tau_x}\left(\hat{\sigma}_{\tau_x} \frac{B_1-B_{\tau_x}}{\sqrt{1-\tau_x}} \leq -\varepsilon^{1/3} \right)
 = \left.\hat{\PP} \left(\xi \leq -\frac{\varepsilon^{1/3}}{\sigma} \right)\right|_{\sigma=\hat{\sigma}_{\tau_x}},
$$
where $\xi$ is a standard normal. As $\hat{\sigma}_{\tau_x}\in[c,C]$, we conclude that the right hand side of the above converges to $1/2$, as $\varepsilon\rightarrow0$, uniformly over almost all random outcomes in $\left\{\tau_x<1\right\}$. In particular, for all small enough $\varepsilon>0$, we have:
%$$
%\bone_{\left\{\tau_x<1 \right\}} \left|\hat{\PP}_{\tau_x}\left(Z_1 > 0\right) - \frac{1}{2} \right| \leq \bone_{\left\{\tau_x<1 \right\}}\frac{1}{4},
%\quad
%\bone_{\left\{\tau_x<1 \right\}} \left|\hat{\PP}_{\tau_x}\left(Z_1 \leq 0\right) - \frac{1}{2} \right| \leq \bone_{\left\{\tau_x<1 \right\}}\frac{1}{4},
%$$
%which yields
$$
\bone_{\left\{\tau_x<1 \right\}} \left|\hat{\PP}_{\tau_x}\left(Y_1 + Z_1 \leq 0\right) - \hat{\PP}_{\tau_x}\left(Y_1 + Z_1 > 0\right) \right| \leq \bone_{\left\{\tau_x<1 \right\}} \delta(\varepsilon)<1,
$$
and, in view of (\ref{eq.necessary.biglemma.step2.1}),
$$
\hat{\PP}(X_1>x) \geq \hat{\EE} \left( \bone_{\left\{ \tau_x < 1 \right\}} \hat{\PP}_{\tau_x}\left(Y_1 + Z_1 \leq 0\right) \right) - \delta(\varepsilon) \hat{\PP}(\tau_x<1)
$$
Summing up the above inequality and (\ref{eq.necessary.biglemma.step2.1}), we obtain
$$
2\hat{\PP}(X_1>x) \geq (1-\delta(\varepsilon))\hat{\PP}(\tau_x<1) = (1-\delta(\varepsilon))\hat{\PP}(\sup_{t\in[0,1]}X_t > x),
$$
which, along with (\ref{eq.necessary.biglemma.step1.res}), yields the statement of the lemma.
\qed
\end{proof}

\section{Proof of Theorem \ref{le:main.zeroTermSpread}}
\label{se:pf.1}

Within the scope of this proof, we adopt the notation introduced in (\ref{eq.xi.not}) and use the following convention.
\begin{notation}\label{not:shift}
The LOB, the bid and ask prices, the expected execution prices, and the demand, are all measured relative to $p^0$. Namely, we use $\nu_n$ to denote $\nu_n\circ (x\mapsto x+p^0_n)^{-1}$, $p^a_n$ to denote $p^a_n-p^0_n$, $p^b_n$ to denote $p^b_n-p^0_n$, $\lambda^a_n$ to denote $\lambda^a_n-p^0_n$, $\lambda^b_n$ to denote $\lambda^b_n-p^0_n$, and $D_n(p)$ to denote $D_n(p^0_n+p)$.
\end{notation}

Herein, we are only concerned with what happens in the last trading period -- at time $(N-1)$, where $N=T/\Delta t$. Hence, we omit the subscript $N-1$ whenever it is clear from the context. In particular, we write $p^a$ and $p^b$ for $p^a_{N-1}$ and $p^b_{N-1}$, $\nu$ for $\nu_{N-1}$, and $\xi$ for $\xi_N$. Note also that, in an LTC equilibrium, we have: $p^a=p^a_N=p^a_{N-1}$, with similar equalities for $p^b$ and $\nu$. For convenience, we also drop the superscript $\Delta t$ in the LOB and the associated bid and ask prices. Finally, we denote by $\tilde{\mathbb{A}}$ the support of a given equilibrium.
%, and by $\mu$ the empirical measure at time $N-1$.
As the roles of $p^a$ and $p^b$ in our model are symmetric, we will only prove the statement of the proposition for $p^b$. 
We are going to show that, under the assumptions of the theorem, there exists a constant $C_0>0$, depending only on the constant $C$ in Assumptions \ref{ass:sigma} and \ref{ass:A.alpha}, such that, for all small enough $\Delta t$, we have, $\PP$-a.s.: 
\begin{equation}\label{eq.prop1.target}
-C_0\leq p^b/\sqrt{\Delta t} < 0
\end{equation}
First, we introduce $\hat{A}^\alpha(p;x)$, which we refer to as the simplified objective: 
\begin{equation}\label{eq.simp.obj.def}
\hat{A}^\alpha(p;x)=\EE^\alpha_{N-1}\left((p-x-\xi)\bone_{\{\xi>p\}}\right).
\end{equation}
Recall that the expected relative profit from posting a limit sell order at price level $p$, in the last time period,\footnote{Recall that everything is measured relative to the  fundamental price, according to the Notational Convention \ref{not:shift}} is given by $A^\alpha(p;p^b_{N})$, where
\begin{equation}\label{eq.true.obj.def}
A^\alpha(p;x)=\EE^\alpha_{N-1}\left((p-x-\xi)\bone_{\{D^+_N(p-\xi)>\nu^+((-\infty,p))\}}\right).
\end{equation} 
The simplified objective is similar to $A^\alpha$, but it assumes that there are no orders posted at better prices than the one posted by the agent.
In particular, $\hat{A}^\alpha(p;x)=A^\alpha(p;x)$ for $p\le p^a$.
Corollary \ref{cor:piecewiseLin}, in Appendix A, states that, in equilibrium, $\PP$-a.s., if the agents in the state $(s,\alpha)$ post limit sell orders, then they post them at a price level $p$ that maximizes the true objective $A^\alpha(p;p^b)$. %However, this optimality may fail o na set of measure zero, and these sets may be different for different states $(s,\alpha)$.
The following lemma shows that the value of the modified objective becomes close to the value of the true objective, for the agents posting limit sell orders close to the ask price.

\begin{lemma}\label{le:simp.to.true.val}
$\PP$-a.s., either $\nu^+(\{p^a\})>0$ or we have:
$$
\left\vert A^{\alpha}(p;p^b) - \hat{A}^{\alpha}(p^a;p^b)\right\vert\to0,
$$
as $p\downarrow p^a$, uniformly over all $\alpha\in\tilde{\mathbb{A}}$.
\end{lemma}
\begin{proof}
If $\nu^+(\{p^a\})=0$, then $\nu^+$ is continuous at $p^a$, and $\nu^+((-\infty,p])\rightarrow0$, as $p\downarrow p^a$. Then, we have
$$
\left| A^{\alpha}(p;p^b) - \hat{A}^{\alpha}(p^a;p^b)\right|
$$
$$
=\left| \EE^{\alpha}_{N-1}\left((p - p^b-\xi)\bone_{\{D^+_N(p-\xi)>\nu^+((-\infty,p))\}}\right) 
- \EE^{\alpha}_{N-1}\left((p^a-p^b-\xi)\bone_{\{\xi>p^a\}}\right)\right|
$$
$$
\le|p-p^a|+\left\Vert p^a-p^b-\xi\right\Vert_{\mathbb{L}^2\left(\PP^{\alpha}_{N-1}\right)}\PP^{\alpha}_{N-1}\left(\xi>p^a,\,D^+_N(p-\xi)\le\nu^+((-\infty,p))\right)
$$
Thus, it suffices to show that: (i) $\left\Vert p^a-p^b-\xi \right\Vert_{\mathbb{L}^2(\PP^{\alpha}_{N-1})}$ is bounded by a finite random variable independent of $\alpha$, and (ii) 
$$
\PP^{\alpha}_{N-1}\left(\xi_N>p^a,\, D^+_N(p-\xi)\le\nu^+((-\infty, p))\right) \to0,
\quad \PP\text{-a.s.},
$$
as $p\downarrow p^a$, uniformly over $\alpha$.
For (i), we have:
$$
\left\Vert p^a-p^b-\xi \right\Vert_{\mathbb{L}^2(\PP^{\alpha}_{N-1})}\le |p^a-p^b|+\left\Vert\xi\right\Vert_{\mathbb{L}^2(\PP^{\alpha}_{N-1})} \leq |p^a-p^b| + 2C \sqrt{\Delta t},
$$
where the constant $C$ appears in Assumptions \ref{ass:sigma} and \ref{ass:A.alpha}.
For (ii), we note that 
$$
\{\xi_N>p^a,\, D^+_N(p-\xi) \le \nu^+((-\infty,p))\}
= \{\xi_N>p^a,\, \xi \le p - D^{-1}_N\left(\nu^+((-\infty,p))\right)\},
$$
as $D_N(\cdot)$ is strictly decreasing, with $D_N(0)=0$. 
Assumption \ref{ass:main.demandInv.unif} implies that 
$$
\kappa^{-1}(\nu^+((-\infty,p)))\leq D^{-1}_N(\nu^+((-\infty,p))) < 0,
$$ 
where $\kappa$ is known at time $N-1$.
Therefore, 
$$
\PP^{\alpha}_{N-1}\left(\xi>p^a,\, D^+_N(p-\xi)\le\nu^+((-\infty, p))\right)
\leq \PP^{\alpha}_{N-1} \left( \xi \in \left(p^a, p - \kappa^{-1}(\nu^+((-\infty,p))) \right] \right).
$$
It remains to show that, $\PP$-a.s., the right hand side of the above converges to zero, uniformly over all $\alpha$. \
Assume that it does not hold. Then, with positive probability $\PP$, there exists $\varepsilon>0$ and a sequence of $(p_k,\alpha_k)$, such that $p_k\downarrow p^a$ and 
$$
\PP^{\alpha_k}_{N-1} \left( \xi \in (p^a, p_k - \kappa^{-1}(\nu^+((-\infty,p_k))) ] \right) \geq \varepsilon.
$$
Notice that, $\PP$-a.s., the family of measures $\left\{ \hat{\mu}_k = \PP^{\alpha_k}_{N-1}\circ \xi^{-1} \right\}_k$ is tight. The latter follows, for example, from the fact that, $\PP$-a.s., the conditional second moments of $\xi$ are bounded uniformly over all $\alpha$ (which, in turn, is a standard exercise in stochastic calculus). Prokhorov's theorem, then, implies that there is a subsequence of these measures that converges weakly to some measure $\hat{\mu}$ on $\RR$. 
Next, notice that, for any fixed $k$ in the chosen subsequence, there exists a large enough $k'$, such that
$$
\left|\hat{\mu} \left( \left(p^a, p_k - \kappa^{-1}(\nu^+((-\infty,p_k))) \right] \right) - \mu_{k'}\left( \left(p^a, p_k - \kappa^{-1}(\nu^+((-\infty,p_k))) \right] \right)\right| \leq \varepsilon/2.
$$
Thus, for any $k$ in the subsequence, we have
$$
\hat{\mu} \left( \left(p^a, p_k - \kappa^{-1}(\nu^+((-\infty,p_k))) \right] \right) \geq \varepsilon/2.
$$
The above is a contradiction, as the intersection of the corresponding intervals, $(p^a, p_k - \kappa^{-1}(\nu^+((-\infty,p_k))) ]$, over all $k$ is empty.
\qed
\end{proof}

Now we are ready to prove the upper bound in (\ref{eq.prop1.target}).

\begin{lemma}\label{bidasksigns}
In any non-degenerate LTC equilibrium, $p^b<0<p^a$, $\PP$-a.s..
\end{lemma}

\begin{proof}
We only show that $p^b<0$ holds, the other inequality being very similar.
%Note that $p^b=\lambda^a_N(\alpha)$ for all $\alpha$. 
Assume that $p^b\ge0$ on some positive $\PP$-probability set $\Omega'\in\mathcal{F}_{N-1}$. We are going to show that this results in a contradiction. 
First, Corollary \ref{cor:piecewiseLin}, in Appendix A, implies that, $\PP$-a.s., if the agents in state $(s,\alpha)$ post a limit sell order, then we must have: $\sup\limits_{p\in\RR} A^{\alpha}(p;p^b) \geq0$. 
In addition, on $\Omega'$, we have: $\hat{A}^{\alpha}(p^a;p^b)<0$ for all $\alpha\in\tilde{\mathbb{A}}$, as $\xi$ has full support in $\RR$ under every $\PP^{\alpha}_{N-1}$ (which, in turn, follows from the fact that $\sigma$ is bounded uniformly away from zero). Then, Lemma \ref{le:simp.to.true.val} implies that there exists a $\mathcal{F}_{N-1}$-measurable $\bar{p}\geq p^a$, such that, on $\Omega'$, the following holds a.s.: if $\nu^+(\{p^a\})=0$ then $\bar{p}>p^a$, and, in all cases,
\begin{equation}\label{eq.bidasksigns.ubopt}
A^{\alpha}(p;p^b) < 0,\,\,\,\,\,\,\,\,\,\forall p\in[p^a, \bar{p}], \,\,\,\,\forall \alpha\in\tilde{\mathbb{A}}
\end{equation}
Clearly, it is suboptimal for an agent to post a limit sell order below $\bar{p}$.
However, an agent's strategy only needs to be optimal up to a set of $\PP$-measure zero, and these sets can be different for different $(s,\alpha)$. Therefore, a little more work is required to obtain the desired contradiction.
Consider the set $B\subset \Omega'\times\mathbb{\RR}\times\tilde{\mathbb{A}}$:
$$
B = \left\{(\omega,s,\alpha)\,|\, \hat{q}(s,\alpha)>0,\,\,\hat{p}(s,\alpha)\leq \bar{p} \right\}.
$$
This set is measurable with respect to $\mathcal{F}_{N-1}\otimes \mathcal{B}\left(\mathbb{\RR}\times\tilde{\mathbb{A}}\right)$, due to the measurability properties of $\hat{q}$ and $\hat{p}$. Notice that, due to the above discussion and the optimality of agents' actions (cf. Corollary \ref{cor:piecewiseLin}, in Appendix A), for any $(s,\alpha)\in\mathbb{\RR}\times\tilde{\mathbb{A}}$, we have:
$$
\PP(\left\{\omega\,|\, (\omega,s,\alpha)\in B \right\}) = 0,
$$
and hence
$$
\EE_{N-1} \int_{\mathbb{\RR}\times\tilde{\mathbb{A}}} \bone_{B}(\omega,s,\alpha) \mu_{N-1}(ds,d\alpha)
= \int_{\mathbb{\RR}\times\tilde{\mathbb{A}}} \EE_{N-1}\left(\bone_{B}(\omega,s,\alpha) \rho_{N-1}(\omega,s,\alpha)\right) \mu^0_{N-1}(ds,d\alpha)
= 0,
$$
where $\rho_{N-1}$ is the Radon-Nikodym density of $\mu_{N-1}$ w.r.t. to the deterministic measure $\mu^0_{N-1}$ (cf. Assumption \ref{ass:dom.mu}).
The above implies that, $\PP_{N-1}$-a.s., $\bone_{B}(\omega,s,\alpha)\rho_{N-1}(\omega,s,\alpha)=0$, for $\mu^0_{N-1}$-a.e. $(s,\alpha)$.
Notice also that, for all $(\omega,s,\alpha)\in \Omega'\times\mathbb{\RR}\times\tilde{\mathbb{A}}$,
$$
\bone_{\left\{\hat{p}(s,\alpha)\leq \bar{p} \right\}} \hat{q}^+(s,\alpha) \bone_{B^c} = 0.
$$
From the above observations and the condition (\ref{eq.nuplus.fixedpoint.def}) in the definition of equilibrium (cf. Definition \ref{def:equil.def}), we conclude that, on $\Omega'$, the following holds a.s.:
$$
\nu^+([p^a,\bar{p}])=0,
$$ 
where $\bar{p}\geq p^a$, and, if $\nu^+(\{p^a\})=0$, then $\bar{p}> p^a$.
This contradicts the definition of $p^a$ (recall that $p^a$ is $\PP$-a.s. finite, due to non-degeneracy of the LOB).
\qed
\end{proof}

It only remains to prove the lower bound on $p^b$ in (\ref{eq.prop1.target}). 
Assume that it does not hold. That is, assume that there exists a family of equilibria, with arbitrary small $\Delta t$, and positive $\PP$-probability $\mathcal{F}_{N-1}$-measurable sets $\Omega^{\Delta t}$, such that $p^b<-C_0\sqrt{\Delta t}$ on $\Omega^{\Delta t}$. We are going to show that this leads to a contradiction with $p^a>0$.
To this end, assume that the agents maximize the simplified objective function, $\hat{A}^{\alpha}$, instead of the true one, $A^{\alpha}$. Then, if $p^b$ is negative enough, the optimal price levels become negative for all $\alpha$. The precise formulation of this is given by the following lemma.

\begin{lemma}\label{gap}
There exists a constant $C_0>0$, s.t., for any small enough $\Delta t$, there exist constants $\epsilon,\delta>0$, s.t., $\PP$-a.s., we have
$$
\hat{A}^{\alpha}(-\delta;x)\ge\epsilon+\sup\limits_{y\ge0}\hat{A}^\alpha(y;x),
$$
for all $\alpha\in\tilde{\mathbb{A}}$ and all $x\le-C_0\sqrt{\Delta t}$.
\end{lemma}
\begin{proof}
Denote $\bar{\xi}=\xi/\sqrt{\Delta t}$ and consider the random function
$$
\bar{A}^\alpha(p;x)=\een\left((p-x-\bar{\xi})\bone_{\{\bar{\xi}>p\}}\right).
$$ 
Notice that 
$$
\hat{A}^\alpha(p;x)=\sqrt{\Delta t}\bar{A}^\alpha\left(p/\sqrt{\Delta t}; x/\sqrt{\Delta t}\right),
$$ 
and, hence, we can reformulate the statement of the lemma as follows: there exists a constant $C_0>0$, s.t., for any small enough $\Delta t$, there exist constants $\epsilon,\delta>0$, s.t., $\PP$-a.s., we have
$$
\bar{A}^{\alpha}(-\delta;x)\ge\epsilon+\sup\limits_{y\ge0} \bar{A}^\alpha(y;x),
$$
for all $\alpha\in\tilde{\mathbb{A}}$ and all $x\le-C_0$.
Notice that
\begin{multline*}
\bar{A}^\alpha(-\delta;x)-\bar{A}^\alpha(y;x)
= -x\een\left(\bone_{\{-\delta<\bar{\xi}\le y\}}\right) - \een\left(\xi\bone_{\{-\delta<\bar{\xi}\le y\}}\right)
- \delta\een\left(\bone_{\{\bar{\xi}>-\delta\}}\right) - y\een\left(\bone_{\{\bar{\xi}>y\}}\right)
\end{multline*}
is non-increasing in $x$, and, hence, such is $\bar{A}^{\alpha}(-\delta;x)-\sup\limits_{y\ge0}\bar{A}^\alpha(y;x)$. Hence, it suffices to prove the above statement for $x=-C_0$.
Next, consider the deterministic function $A_\sigma(p;x)$, defined via 
\begin{equation}\label{eq.Asigma.def}
A_\sigma(p;x)=\hat{\EE}\left((p-x-\sigma\eta_0)\bone_{\{\sigma\eta_0>p\}}\right),
\end{equation}
where $\eta_0$ is a standard normal random variable on some auxiliary probability space $(\hat{\Omega},\hat{\PP})$.  It follows from Lemma \ref{gapproxapplied} that there exists a function $\varepsilon_2(\cdot)\ge0$, s.t. $\varepsilon_2(\Delta t)\to0$, as $\Delta t\to0$, and, $\PP$-a.s., we have
$$
\left|\bar{A}^\alpha(p;-C_0)-A_{\sigma_{t_{N-1}}}(p;-C_0)\right|\le\varepsilon_2(\Delta t),
$$
for all $\alpha\in\tilde{\mathbb{A}}$ and all $p\in\RR$. Then, as we can always choose $\Delta t$ small enough, so that $\varepsilon_2(\Delta t)<\epsilon$, the statements of the lemma would follow if we can show that there exist constants $\epsilon,\delta,C_0>0$, s.t., $\PP$-a.s.,
$$
A_{\sigma_{t_{N-1}}}(-\delta;-C_0)\ge3\epsilon+\sup\limits_{y\ge0}A_{\sigma_{t_{N-1}}}(y;-C_0)
$$
As $\sigma_{t_{N-1}}(\omega)\in[1/C,C]$, $\PP$-a.s., it suffices to find $\epsilon,\delta,C_0>0$, s.t.
$$
A_\sigma(-\delta;-C_0)\ge3\epsilon+\sup\limits_{y\ge0}A_\sigma(y;-C_0), \quad\forall\,\sigma\in[1/C,C].
$$
Note that the above inequality does not involve $\omega$ or $\xi$, and it is simply a property of a deterministic function. 
Notice also that $A_\sigma(p;x)=\sigma A_1\left(p/\sigma;x/\sigma\right)$, with $A_1$ given in (\ref{eq.Asigma.def}).
%where 
%$$
%A_1(p;x)=\EE\left[(p-x-\eta_0)\bone_{\{\eta_0>p\}}\right],
%$$ 
%with $\eta_0$ being a standard normal. 
Then, if we denote by $F(x)$ and $f(x)$, respectively, the cdf and pdf of a standard normal, we obtain
$$
A_1(p;x)=(p-x)(1-F(p))-\int_p^{\infty} t f(t)\text{d}t.
$$ 
A straightforward calculation gives us the following useful properties of $A_1$ and $A_\sigma$

\par\noindent(i) For any $\sigma>0$ and any $x<0$, the function $p\mapsto A_\sigma(p;x)$ has a unique maximizer $p_\sigma(x)$, in particular, it is increasing in $p\le p_\sigma(x)$ and decreasing in $p\ge p_\sigma(x)$.

\par\noindent(ii) The function 
$$
x\mapsto p_\sigma(x)=\sigma p_1(x/\sigma)=\sigma\left((1-F)/f\right)^{-1}(-x/\sigma)
$$ 
is increasing in $x<0$ and converges to $-\infty$, as $x\to-\infty$.

Then, choosing $C_0$ large enough, so that $p_1(-C_0/C)<0$, ensures $p_\sigma(-C_0)<0$, for all $\sigma\in[1/C,C]$. Setting $\delta=-p_1(-C_0/C)/C$ guarantees that $p_\sigma(-C_0)\le-\delta$, for all $\sigma\in[1/C,C]$. Then, by property (i) above, we have, for all $\sigma\in[1/C,C]$
$$
A_\sigma(-\delta;-C_0)>A_\sigma(0;-C_0)=\sup\limits_{y\ge0}A_\sigma(y;-C_0).
$$
Finally, as $A_\sigma(-\delta;-C_0)-A_\sigma(0;-C_0)$ is a continuous function of $\sigma\in[1/C,C]$, we can find $\epsilon$, such that
$$
A_\sigma(-\delta;-C_0)\ge3\epsilon+\sup\limits_{y\ge0}A_\sigma(y;-C_0), \quad\forall\,\sigma\in[1/C,C].
$$
\qed
\end{proof}

Recall that our assumption is that $p^b<-C_0\sqrt{\Delta t}$ holds on a set $\Omega^{\Delta t}$ of positive $\PP$-measure. Recall also that $p^a>0$, $\PP$-a.s., due to Lemma \ref{bidasksigns}.
Then, Lemmas \ref{le:simp.to.true.val} and \ref{gap} imply that there exists $\mathcal{F}_{N-1}$-measurable $\bar{p}\geq p^a$, s.t., on $\Omega^{\Delta t}$, we have a.s.: if $\nu^+(\{p^a\})=0$ then $\bar{p}>p^a$, and, in all cases,
$$
A^{\alpha}(p;p^b) < \sup_{p'\in\RR} A^{\alpha}(p';p^b),\,\,\,\,\,\,\,\,\forall p\in[p^a,\bar{p}],\,\,\,\,\forall \alpha\in\tilde{\mathbb{A}}.
$$
It is intuitively clear that posting limit sell orders at the above price levels $p$ must be suboptimal for the agents. However, the above inequality, on its own, does not yield a contradiction, as the agents' strategies are only optimal up to a set of $\PP$-probability zero, and these sets may be different for different states $(s,\alpha)$. To obtain a contradiction with the definition of $p^a$, we simply repeat the last part of the proof of Lemma \ref{bidasksigns} (following equation (\ref{eq.bidasksigns.ubopt})). This ensures that (\ref{eq.prop1.target}) holds and completes the proof of the theorem.

\section{Proof of Theorem \ref{thm:main.necessary}}
\label{se:pf.2}

Within the scope of this proof, we adopt the notation introduced in (\ref{eq.xi.not}) and use Notational Convention \ref{not:shift} (i.e. we measure the LOB, the expected execution prices, and the demand, relative to $p^0$, but keep the same variables' names).
Assume that the statement of the theorem does not hold: i.e., there exists $\alpha_0\in\tilde{\mathbb{A}}$, such that $\tilde{p}^0$ is not a martingale under $\PP^{\alpha_0}$. Then, there exists $s\in[0,T)$, s.t., with positive probability $\PP^{\alpha_0}$, we have
$$
\tilde{\EE}^{\alpha_0}_{s}\tilde{p}^0_T \neq \tilde{p}^0_{s}.
$$ 
Without loss of generality, we assume that there exists a constant $\delta>0$ and a set $\Omega'\in\mathcal{F}_{s}$, having positive probability $\PP^{\alpha_0}$ (and hence $\PP$), s.t., for all random outcomes in $\Omega'$, we have 
\begin{equation}\label{eq.thm1.delta.def}
\tilde{\EE}^{\alpha_0}_{s}(\tilde{p}^0_T - \tilde{p}^0_{s})\geq \delta
\end{equation}
(the case of negative values is analogous). 
%Then, we can find $\Omega'\in\mathcal{F}_t$ with positive $\PP^{\alpha_0}$ (hence $\PP$) probability such that, for \textit{every} $\omega\in\Omega'$, have (i) $\EE^{\alpha_0}_t[p^0_T-p^0_t]<0$, (ii) conclusions of the corollary \ref{cor:piecewiseLin} hold (for all $\Delta t$ in our sequence) and (iii) conclusions of the proposition \ref{le:main.zeroTermSpread} hold, that is, 
Next, we fix an arbitrary $\Delta t$ from a given family and consider the associated non-degenerate LTC equilibrium.

\begin{lemma}\label{le:necessary.1}
There exists a deterministic function $\varepsilon(\cdot)\geq0$, s.t. $\varepsilon(\Delta t)\rightarrow0$, as $\Delta t\rightarrow0$, and, for any small enough $\Delta t>0$, there exists $n=0,\ldots,N - 3$ and $\Omega''\in\mathcal{F}_n$, s.t. $\PP^{\alpha_0}_n(\Omega'')>0$ and the following holds on $\Omega''$
$$
\PP^{\alpha_0}_{n+2} \left( \EE^{\alpha_0}_{n+3} \left(p^0_N - p^0_{n+3} \right) \leq \delta/2\right) \leq \varepsilon(\Delta t).
$$
\end{lemma}
\begin{proof}
The proof follows from Assumption \ref{ass:main.mu.cont.strong}. Consider $t=t'=s$ and $t'' = t_{n+2}$. Then, Assumption \ref{ass:main.mu.cont.strong} implies
$$
\tilde{\PP}^{\alpha_0}_{s} \left( \left|\tilde{\EE}^{\alpha_0}_{t_{n+2}} \int_{s}^T \mu^{\alpha_0}_u du - \tilde{\EE}^{\alpha_0}_{s} \int_{s}^T \mu^{\alpha_0}_u du\right| \geq \varepsilon(\Delta t)\right) \leq \varepsilon(\Delta t)
$$
on $\Omega'$, a.s.. 
Notice also that 
$$
\tilde{\EE}^{\alpha_0}_{s} (\tilde{p}^0_T - \tilde{p}^0_{s}) 
= \tilde{\EE}^{\alpha_0}_{s} \int \limits_{s}^T \mu^{\alpha_0}_u \text{d}u.
$$
Then, assuming that $\varepsilon(\Delta t)$ is small enough and recalling (\ref{eq.thm1.delta.def}), we obtain 
$$
\tilde{\PP}^{\alpha_0}_{s} \left( \tilde{\EE}^{\alpha_0}_{t_{n+2}} \int_{s}^T \mu^{\alpha_0}_u du \leq 3\delta/4 \right) \leq \varepsilon(\Delta t),
$$
on $\Omega'$.
Therefore, there exists a set $\Omega''\in\mathcal{F}_{s}\subset\mathcal{F}_{t_{n}}$, s.t. $\tilde{\PP}^{\alpha_0}_{t_{n}}(\Omega'')>0$ and
$$
\tilde{\EE}^{\alpha_0}_{t_{n+2}} \int_{s}^T \mu^{\alpha_0}_u du \geq 3\delta/4,
$$
on $\Omega''$. 
%Clearly, the same inequality holds if $n$ is replaced by $n+1$, provided $\Delta t$ is small enough. 
Next, we choose $t=s$, $t'=t_{n+2}$, $t'' = t_{n+3}$, and use Assumption \ref{ass:main.mu.cont.strong}, to obtain
$$
\tilde{\PP}^{\alpha_0}_{t_{n+2}} \left( \left|\tilde{\EE}^{\alpha_0}_{t_{n+3}} \int_{s}^T \mu^{\alpha_0}_u du - \tilde{\EE}^{\alpha_0}_{t_{n+2}} \int_{s}^T \mu^{\alpha_0}_u du\right| \geq \varepsilon(\Delta t)\right) \leq \varepsilon(\Delta t),
$$
on $\Omega''$, a.s.. 
Assuming that $\varepsilon(\Delta t)$ is small enough and using the last two inequalities, we obtain
$$
\tilde{\PP}^{\alpha_0}_{t_{n+2}} \left( \tilde{\EE}^{\alpha_0}_{t_{n+3}} \int_{s}^T \mu^{\alpha_0}_u du \leq \delta/2 \right) \leq \varepsilon(\Delta t).
$$
Finally, due to Assumption \ref{ass:A.alpha}, and as $\Delta t$ is small, we can replace $\int_{s}^T \mu^{\alpha_0}_u du$ by $\int_{t_{n+3}}^T \mu^{\alpha_0}_u du$, and $\delta/2$ by $\delta/4$, in the above equation.
%And, as before, we notice that the same procedure can be repeated using $n+1$ in place of $n$. 
This completes the proof of the lemma.
\qed
\end{proof}

Using the strategy at which the agent in state $(1,\alpha_0)$ waits until the last moment $n=N$, we conclude that the process $(\lambda^a_n(\alpha_0) + p^0_n)$ must be a supermartingale under $\PP^{\alpha_0}$. More precisely, due to the definition of an optimal strategy, we have, $\PP$-a.s.
$$
\lambda^a_{n+2}(\alpha_0) \geq \EE^{\alpha_0}_{n+2} \lambda^a_N(\alpha_0) + \EE^{\alpha_0}_{n+2}\left( \EE^{\alpha_0}_{n+3}(p^0_N - p^0_{n+3}) + \xi_{n+3} \right).
$$
Recall that $\lambda^a_N(\alpha_0) = p^b_N$ and, due to Theorem \ref{le:main.zeroTermSpread} (more precisely, it follows from the proof of the theorem), there exists a constant $C_0>0$, s.t., for all small enough $\Delta t>0$, the following holds $\PP$-a.s.
$$
-C_0\sqrt{\Delta t}\le p^{b}_N <0<p^{a}_N\le C_0\sqrt{\Delta t}.
$$
Thus, we have, $\PP$-a.s.
\begin{equation}\label{eq.necessary.lambdab.est.1}
\lambda^a_{n+2}(\alpha_0) \geq -C_0\sqrt{\Delta t} 
+ \EE^{\alpha_0}_{n+2}\left( \EE^{\alpha_0}_{n+3}(p^0_N - p^0_{n+3}) \right)
+  \EE^{\alpha_0}_{n+2} \xi_{n+3}.
\end{equation}
Due to Assumption \ref{ass:A.alpha}, we have, $\PP$-a.s.
$$
\EE^{\alpha_0}_{n+2} \xi_{n+3} \leq C\Delta t,
\,\,\,\,\,\,\,\,\,\,\,\,\,\,\,
\left|\EE^{\alpha_0}_{n+3}(p^0_N - p^0_{n+3})\right|
\leq C T,
$$
and, hence,
$$
\lambda^a_{n+2}(\alpha_0) \geq -C_0\sqrt{\Delta t} + CT + C\Delta t.
$$
In addition, making use of Lemma \ref{le:necessary.1}, we conclude that, for any small enough $\Delta t$, there exist $n=0,\ldots,N - 2$ and $\Omega''\in\mathcal{F}_n$, s.t. $\PP^{\alpha_0}_n(\Omega'')>0$
%$$
%\PP^{\alpha_0}_n(\Omega''')>1-\varepsilon(\Delta t)\,\,\,\,\,\,\,\text{on}\,\,\Omega'',
%$$
and
$$
\PP^{\alpha}_{n+2} \left( \EE^{\alpha}_{n+3} \left(p^0_N - p^0_{n+3} \right) \leq \delta/2\right) \leq \varepsilon(\Delta t),\,\,\,\,\,\text{on}\,\,\Omega''.
$$
Using (\ref{eq.necessary.lambdab.est.1}) and assuming that $\Delta t$ is small enough, we obtain
$$
\lambda^a_{n+2}(\alpha_0) \geq \delta/4,\,\,\,\,\,\,\,\text{on}\,\,\Omega''.
$$
Next, Corollary \ref{cor:piecewiseLin}, in Appendix A, implies that, $\PP$-a.s.,
$$
p^b_{n+1}\geq \EE^{\alpha_0}_{n+1}\left(\lambda^a_{n+2}(\alpha_0)+\xi_{n+2}\big\vert \xi_{n+2}<p^b_{n+1}\right).
$$
Thus, on $\Omega''$, we obtain
\begin{equation}\label{eq.necessary.pa.est.1}
p^b_{n+1} - \EE^{\alpha_0}_{n+1}\left(\xi_{n+2}\big\vert\xi_{n+2}<p^b_{n+1}\right) \geq \delta /4.
\end{equation}

The following lemma shows that, for any number $p$, the conditional expectation of the fundamental price increment, $\EE^{\alpha_0}_{n+1}(\xi_{n+2}|\xi_{n+2}<p)$, approaches $p$ as the size of the time interval vanishes. This result follows from Lemma \ref{le:necessary.marginal.maximum}.

\begin{lemma}\label{le:necessary.2}
There exists a constant $C_3>0$, s.t., for all small enough $\Delta t>0$, and for any $t\in[0,T-\Delta t]$, the following holds $\PP$-a.s.
$$
\sup\limits_{p\le0}\left| p - \tilde{\EE}^{\alpha_0}_t\left(\tilde{p}^0_{t+\Delta t} - \tilde{p}^0_t\,\big\vert\, \tilde{p}^0_{t+\Delta t}-\tilde{p}^0_t<p\right) \right| \leq C_3 \sqrt{\Delta t}.
$$
\end{lemma}
\begin{proof}
Fix $t$ and $\Delta t>0$ and consider the evolution of $\tilde{p}^0_s$, for $s\in[t,t+\Delta t]$, under $\PP^{\alpha_0}_t$
$$
\tilde{p}^0_{s} - \tilde{p}^0_t = \int_t^s \mu^{\alpha_0}_u du + \int_t^s \sigma_u dW^{\alpha_0}_u,
$$
where $W^{\alpha_0}$ is a Brownian motion under $\PP^{\alpha_0}$.
Rescaling by $\sqrt{\Delta t}$, we obtain
$$
(\tilde{p}^0_{s} - \tilde{p}^0_t)/\sqrt{\Delta t} = X_{(s-t)/\Delta t},
\quad X_s = \int_0^s \hat{\mu}_u du + \int_0^s \hat{\sigma}_u d\hat{W}_u,
\quad s\in[0,1],
$$
with
$$
\hat{\mu}_s = \sqrt{\Delta t} \, \mu^{\alpha_0}_{t+s\Delta t},
\quad \hat{\sigma}_s = \sigma_{t+s\Delta t},
\quad \hat{W}_s = \frac{1}{\sqrt{\Delta t}} \left(W^{\alpha_0}_{t+s\Delta t} - W^{\alpha_0}_t\right),
\quad s\in[0,1].
$$
Notice that the above processes are adapted to the new filtration $\hat{\mathbb{F}}$, with $\hat{\mathcal{F}}_s = \tilde{\mathcal{F}}_{t+s\Delta t}$, and, $\PP$-a.s., under $\tilde{\PP}^{\alpha_0}_t$, $\hat{W}$ is a Brownian motion with respect to $\hat{\mathbb{F}}$. 
%In particular, the above stochastic integral is well defined.
Next, due to Assumptions \ref{ass:sigma} and \ref{ass:main.L2.strong}, for any small enough $\Delta t>0$, $\PP$-a.s., the dynamics of $(-X_s)$, under $\tilde{\PP}^{\alpha_0}_t$, satisfy all the assumptions of Lemma \ref{le:necessary.marginal.maximum}. As a result, we obtain
$$
\tilde{\PP}^{\alpha_0}_t(X_1 < -x-z) 
\leq C_1 e^{-z} \tilde{\PP}^{\alpha_0}_t(X_1 < -x),
\,\,\,\,\,\,\,\,\,\,\,\,\,\forall x,z\geq0.
$$
Finally, we notice that
$$
\sup\limits_{p\le0}\left| p - \tilde{\EE}^{\alpha_0}_t\left(\tilde{p}^0_{t+\Delta t} - \tilde{p}^0_t\big\vert \tilde{p}^0_{t+\Delta t}-\tilde{p}^0_t<p\right)\right|
 = \sqrt{\Delta t} \sup\limits_{p\le0}\left| p - \tilde{\EE}^{\alpha_0}_t\left(X_1\big\vert X_1<p\right)\right|
$$
$$
= \sqrt{\Delta t} \sup\limits_{p\le0}\left| p - \frac{\int_{-p}^{\infty} x \,d\,\tilde{\PP}^{\alpha_0}_t(X_1 < -x) }{\tilde{\PP}^{\alpha_0}_t(X_1 < p)} \right|
= \sqrt{\Delta t} \sup\limits_{p\le0}\left| \frac{\int_{0}^{\infty} \tilde{\PP}^{\alpha_0}_t(X_1 < p - z) dz}{\tilde{\PP}^{\alpha_0}_t(X_1 < p)} \right|
\leq C_1 \sqrt{\Delta t},
$$
%$$
%= \sqrt{\Delta t} \sup\limits_{p\le0}\left| p -  \frac{p \tilde{\PP}^{\alpha_0}_t(X_1 < p) - \int_{-p}^{\infty} \tilde{\PP}^{\alpha_0}_t(X_1 < -x) dx }{\tilde{\PP}^{\alpha_0}_t(X_1 < p)}\right|
%$$
%$$
%= \sqrt{\Delta t} \sup\limits_{p\le0}\left| \frac{\int_{0}^{\infty} \tilde{\PP}^{\alpha_0}_t(X_1 < p - z) dz}{\tilde{\PP}^{\alpha_0}_t(X_1 < p)} \right|
%\leq C_1C_2 \sqrt{\Delta t}
%$$
which completes the proof of the lemma.
\qed
\end{proof}

Using (\ref{eq.necessary.pa.est.1}) and Lemma \ref{le:necessary.2}, we conclude that, for all small enough $\Delta t$, we have: $p^b_{n+1} > 0$ on $\Omega''$, $\PP$-a.s..
%\begin{equation}\label{eq.necessary.pa.est.1}
%p^b_{n+1} > 0
%\end{equation}
In addition, Corollary \ref{cor:piecewiseLin}, in Appendix A, implies that, for any $\alpha\in\tilde{\mathbb{A}}$, the following holds $\PP$-a.s.
$$
\lambda^a_{n+1}(\alpha) \geq p^b_{n+1}.
$$
Next, with a slight abuse of notation (similar notation was introduced in the proof of Proposition \ref{le:main.zeroTermSpread}), we consider the simplified objective of an agent who posts a limit sell order at the ask price $p^a_n$
$$
\hat{A}^{\alpha}(p^a_n;\lambda^a_{n+1}) = \EE^{\alpha}_n\left( p^a_n - \lambda^a_{n+1} - \xi_{n+1} \,|\, \xi_{n+1} > p^a_n \right).
$$
The above estimates imply that, on $\Omega''$, we have, $\PP$-a.s.
\begin{equation}\label{eq.necessary.simpObj.neg}
\hat{A}^{\alpha}(p^a_n;\lambda^a_{n+1}) \leq  
\EE^{\alpha}_n\left( p^a_n - \xi_{n+1} \,|\, \xi_{n+1} > p^a_n \right)
- \EE^{\alpha}_n\left( p^b_{n+1}\bone_{\Omega''} \,|\, \xi_{n+1} > p^a_n \right)
< 0,\,\,\,\,\,\,\,\,\,\,\,\forall \alpha\in\tilde{\mathbb{A}}.
\end{equation}
To obtain the last inequality in the above, we recall that $\Omega''\in\mathcal{F}_n$ and, $\PP$-a.s., $\bone_{\Omega''}\PP_n(\Omega\setminus\Omega'')=0$, $p^b_{n+1} > 0$ on $\Omega''$, and $\PP^{\alpha}_n(\xi_{n+1} > p^a_n)>0$, for all $\alpha\in\tilde{\mathbb{A}}$.
Next, repeating the proof of Lemma \ref{le:simp.to.true.val} (and using the fact that $\lambda^a_{n+1}$ is absolutely bounded, as shown in Corollary \ref{prop:main.smallspread}), we conclude that, $\PP$-a.s., either $\nu^+_n(\{p^a_n\})>0$, or we have
$$
\left\vert A^{\alpha}(p;\lambda^a_{n+1}) - \hat{A}^{\alpha}(p^a_n;\lambda^a_{n+1})\right\vert\to0,
$$
as $p\downarrow p^a$, uniformly over all $\alpha\in\tilde{\mathbb{A}}$, where we introduce the true objective,
\begin{equation*}
A^\alpha(p;\lambda^a_{n+1})=\EE^\alpha_{n}\left(\left(p-\lambda^a_{n+1}-\xi_{n+1}\right)\bone_{\{D^+_{n+1}(p-\xi_{n+1})>\nu^+_n((-\infty,p))\}}\right).
\end{equation*} 
This convergence, along with (\ref{eq.necessary.simpObj.neg}), implies that there exists a $\mathcal{F}_{n}$-measurable $\bar{p}\geq p^a_n$, such that, on $\Omega''$, the following holds $\PP$-a.s.: if $\nu^+_n(\{p^a_n\})=0$ then $\bar{p}>p^a_n$, and, in all cases,
\begin{equation*}
A^{\alpha}(p;\lambda^a_{n+1}) < 0,\,\,\,\,\,\,\,\,\,\forall p\in[p^a_n, \bar{p}], \,\,\,\,\forall \alpha\in\tilde{\mathbb{A}}.
\end{equation*} 
Finally, we repeat the last part of the proof of Lemma \ref{bidasksigns} (following equation (\ref{eq.bidasksigns.ubopt})), to obtain a contradiction with the definition of $p^a_n$, and complete the proof of the theorem.
%The above proof also yields Corollary \ref{cor:fundprice.midpoint}.
The last argument also shows that, when $\Delta t$ is small enough, it becomes suboptimal for the agents to post limit sell orders, as the expected relative profit from this action becomes negative, causing the market to degenerate.

\section{Summary and future work}
\label{se:conclusion}

In this paper, we present a new framework for modeling market microstructure, which does not require the existence of a designate market maker, and in which the LOB arises endogenously, as a result of equilibrium between multiple strategic players (aka agents).
This framework is based on a continuum-player game. It closely approximates the mechanics of an auction-style exchange, so that, in particular, it can be used to analyze the liquidity effects of changes in the rules of the exchange.
We use the proposed framework to study the liquidity effects of high trading frequency.
In particular, we demonstrate the dual nature of high trading frequency. On the one hand, in the absence of a bullish or bearish signal about the asset, the higher trading frequency improves the efficiency of the market. On the other hand, at a sufficiently high trading frequency, even a very small trading signal may amplify the adverse selection effect, creating a disproportionally large change in the LOB, which is interpreted as an endogenous liquidity crisis.

The present article raises many questions for further research.
%In particular, it is desirable to model more carefully the internal market orders of the agents. In the current setting, the agents form beliefs about the total order flow, consisting  ignore the potential market orders of other agents, when predicting the future demand, and expect execution from the external orders only. Such setting is consistent if, in equilibrium, no agents choose to submit market orders (which is the case in the example of an equilibrium that we construct herein). However, in general, it would be more realistic to incorporate the anticipation of internal market orders in the agent's models for future demand.
Notice that our main results are of a qualitative nature: they demonstrate the general behavior of the LOB, as a function of trading frequency, but do not immediately allow for computations.
It would also be interesting to establish quantitative results.
In particular, we would like to construct an equilibrium in a more realistic, and more concrete, model than the one used in Section \ref{se:examples}. Such a model would allow for heterogeneous beliefs, and it would prescribe the specific sources of information (i.e., relevant market factors) used by the agents to form their beliefs. A model of this type could be calibrated to market data and used to study the effects of changes in relevant market parameters on the LOB. 
Finally, it would be interesting to develop a continuous time version of the proposed framework, in order to better capture the present state of the markets, where the trading frequency is not restricted. All these questions are the subject of our follow-up paper \cite{GaydukNadtochiy2}.

\section{Appendix A}
%\section{Representation of the value function}
%\label{se:representation}

This section contains several useful technical results on the representation of the value function of an agent in the proposed game. 
%These results follow from the classical Dynamic programming Principle, adjusted to the present setting.
%Let us now consider the (stochastic) value function of an agent for a fixed $(m,s,\alpha,\nu)$: 
%\begin{equation}\label{eq.gen.Val.randField}
%V^{\nu}_m(s,\alpha) = \text{esssup}_{p,q,r} J^{(p,q,r)}\left(m,s,\alpha,\nu\right),
%\end{equation}
%where the essential supremum is taken under $\PP$, over all admissible controls $(p,q,r)$, and $J^{(p,q,r)}$ is given by (\ref{eq.intro.Jm.def}). 
%It is clear that, if an optimal control exists, then the value function is well defined.
Notice that (\ref{eq.stateProc.def}) and (\ref{eq.intro.Jm.def}) imply that, if $\nu$ is admissible, then, for any $(\alpha,m,p,q,r)$, we have, $\PP$-a.s.
$$
\left|J^{(p,q,r)}\left(m,s,\alpha,\nu\right) - J^{(p,q,r)}\left(m,s',\alpha,\nu\right)\right|
\leq |s-s'|\, \EE^{\alpha}_m |p^a_N| \vee |p^b_N|,
\,\,\,\,\,\,\,\,\,\,\forall s,s'\in\mathbb{\RR}
$$
This implies that every $J^{(p,q,r)}\left(m,\cdot,\alpha,\nu\right)$ and $V^{\nu}_m(\cdot,\alpha)$ has a continuous modification under $\PP$.
Thus, whenever $\nu$ is admissible, we define the value function of an agent as the aforementioned continuous modification of the left hand side of (\ref{eq.gen.Val.randField}).

\begin{lemma}\label{le:DPP}
Assume that an optimal control exists for an admissible LOB $\nu$.
Assume also that, for any $\alpha\in\mathbb{A}$, the associated value function $V^{\nu}_n(\cdot,\alpha)$, defined in (\ref{eq.gen.Val.randField}), is measurable with respect to $\mathcal{F}_n\otimes\mathcal{B}(\RR)$.
Then, it satisfies the following Dynamic Programming Principle.
\begin{itemize}

\item For $n=N$ and all $(s,\alpha)\in\mathbb{S}$, we have, $\PP$-a.s.
\begin{equation}\label{eq.het.VN}
V^{\nu}_N(s,\alpha) = s^+ p^b_N - s^- p^a_N
\end{equation}

\item For all $n=N-1,\ldots,0$ and all $(s,\alpha)\in\mathbb{S}$, we have
$$
V^{\nu}_n(s,\alpha) = \text{esssup}_{p,q,r}\left\{\bone_{\left\{r_n=0\right\}}\EE_n^{\alpha} \left( V^{\nu}_{n+1}\left(s,\alpha\right) 
+ \left(q_n p_n + V^{\nu}_{n+1}\left(s-q_n,\alpha\right) - V^{\nu}_{n+1}\left(s,\alpha\right)\right)\cdot
\right.\right.
$$
\begin{equation}\label{eq.het.Vn}
\left.\left.
\cdot\left( \bone_{\left\{q_n\geq0,\,D^+_{n+1}(p_n) > \nu^+_n((-\infty,p_n)) \right\}} + \bone_{\left\{q_n<0,\,D^-_{n+1}(p_n) > \nu^-_n((p_n,\infty)) \right\}} \right) \right)\right.
\end{equation} 
$$
\left.
+ \bone_{\left\{r_n=1\right\}} \left( q^+_n p^b_n - q^-_n p^a_n + \EE_n^{\alpha} V^{\nu}_{n+1}\left(s-q_n,\alpha\right) \right)
\right\},
$$
where the essential supremum is taken under $\PP$, over all admissible controls $(p,q,r)$. 
%$\mathcal{F}_n$-measurable random variables $p,q,r$, such that $r$ takes values in $\left\{0,1\right\}$, while $p$ and $q$ take values in $\RR$. 
%Any triplet $(p,q,r)$ that attains the supremum can the chosen as an optimal control.
\end{itemize}
\end{lemma}
\begin{proof}
The most important step is to show that, for all $n=0,\ldots N-1$ and $(s,\alpha)\in\mathbb{S}$,
\begin{equation}\label{eq.DPP.aux.1}
V^{\nu}_n(s,\alpha) = \text{esssup}_{p,q,r} 
\EE^{\alpha}_n \left( V^{\nu}_{n+1}\left(S^{n,s,(p,q,r)}_{n+1},\alpha\right) 
- g^{\nu}_{n}\left(p_n,q_n,r_n,D_{n+1}\right) \right),
\end{equation}
where the essential supremum is taken under $\PP$, over all admissible controls $(p,q,r)$, and
$$
g^{\nu}_{n}\left(p_n,q_n,r_n,D_{n+1}\right) = \left(p_n\bone_{\left\{ r_n = 0\right\}} + p^a_n\bone_{\left\{ r_n = 1, q_n <0\right\}} + p^b_n\bone_{\left\{ r_n = 1, q_n >0\right\}} \right) \Delta S^{n,s,(p,q,r)}_{n+1}
$$
does not depend on $s$.
Assume that $J^{(p,q,r)}\left(n,\cdot,\alpha,\nu\right)$ is a continuous modification of the objective function.
Notice that, for all $m\leq k \leq n$, we have, $\PP$-a.s.
$$
\EE^{\alpha}_k J^{(p,q,r)}\left(n,S_n^{m,s,(p,q,r)},\alpha,\nu\right)
= J^{(p,q,r)}\left(k,S_k^{m,s,(p,q,r)},\alpha,\nu\right)
+ \EE^{\alpha}_k \sum_{j=k}^{n-1} g^{\nu}_{j}\left(p_j,q_j,r_j,D_{j+1}\right)
$$
Notice also that, for any $(p,q,r)$ we have, $\PP$-a.s.: $J^{(p,q,r)}\left(m,s,\alpha,\nu\right) \leq V_m^\nu(s,\alpha)$, for all $s\in\mathbb{S}$.
Let us show that the left hand side of (\ref{eq.DPP.aux.1}) is less than its right hand side
$$
V_m^\nu(s,\alpha)
=\text{essup}_{p,q,r} J^{(p,q,r)}\left(m,S_m^{m,s,(p,q,r)},\alpha,\nu\right)
$$
$$
=\text{essup}_{p,q,r} \EE^{\alpha}_m \left( J^{(p,q,r)}\left(m+1,S_{m+1}^{m,s,(p,q,r)},\alpha,\nu\right)
- g^{\nu}_{m}\left(p_m,q_m,r_m,D_{m+1}\right) \right)
$$
$$
\leq \text{essup}_{p,q,r} \EE^{\alpha}_m \left( V^{\nu}_{m+1}\left(S_{m+1}^{m,s,(p,q,r)},\alpha\right)
- g^{\nu}_{m}\left(p_m,q_m,r_m,D_{m+1}\right) \right)
$$
Next, we show that the right hand side of (\ref{eq.DPP.aux.1}) is less than its left hand side.
For any $(p,q,r)$, we have, $\PP$-a.s.
\begin{equation*}
\EE^\alpha_m \left(V^\nu_{m+1}\left(S_{m+1}^{m,s,(p,q,r)},\alpha\right)
- g^{\nu}_{m}\left(p_m,q_m,r_m,D_{m+1}\right)\right)
\end{equation*}
$$
=\EE^\alpha_m \left(J^{(\hat{p},\hat{q},\hat{r})}\left(m+1,S_{m+1}^{m,s,(p,q,r)},\alpha,\nu\right)
- g^{\nu}_{m}\left(p_m,q_m,r_m,D_{m+1}\right)\right)
= J^{(\tilde{p},\tilde{q},\tilde{r})}\left(m,s,\alpha,\nu\right)
\leq V^{\nu}_m(s,\alpha),
$$
where $(\tilde{p}_n,\tilde{q}_n,\tilde{r}_n)$ coincide with $(\hat{p}_n,\hat{q}_n,\hat{r}_n)$, for $n\geq m+1$, while they are equal to $(p_m,q_m,r_m)$, for $n=m$.
The proof is completed easily by plugging the dynamics of the state process, (\ref{eq.stateProc.def}), into (\ref{eq.DPP.aux.1}).
\qed
\end{proof}

The following corollary provides a more explicit recursive formula for the value function and optimal control. In particular, it states that the value function of an agent at any time remains linear in $s$, in both positive and negative half lines (with possibly different slopes).

\begin{cor}\label{cor:piecewiseLin}
Assume that an admissible LOB $\nu$ has an optimal control $(\hat{p},\hat{q},\hat{r})$. Then, for any $(s,\alpha)\in\mathbb{S}$, the following holds $\PP$-a.s., for all $n=0,\ldots,N-1$
\begin{enumerate}

\item $V^{\nu}_n(s,\alpha) = s^+ \lambda^a_n(\alpha) - s^- \lambda^b_n(\alpha)$, with some adapted processes $\lambda^a(\alpha)$ and $\lambda^b(\alpha)$, such that $\lambda^a_N(\alpha) = p^b_N$ and $\lambda^b_N(\alpha) = p^a_N$;

\item $p^a_n \geq \EE_{n}^{\alpha} \left( \lambda^a_{n+1}(\alpha) \right)$ and $p^b_n \leq \EE_{n}^{\alpha} \left( \lambda^b_{n+1}(\alpha) \right)$;

\item if, for some $p\in\RR$, $\PP^{\alpha}_n\left(D^+_{n+1}(p) > \nu^+_{n}((-\infty,p))\right)>0$, then
\begin{equation*}\label{eq.het.NoPred.1}
p \leq \EE_{n}^{\alpha} \left( \lambda^b_{n+1}(\alpha) \,|\, D^+_{n+1}(p) > \nu^+_{n}((-\infty,p)) \right);
\end{equation*}

\item if, for some $p\in\RR$, $\PP^{\alpha}_n\left(D^-_{n+1}(p) > \nu^-_{n}((p,\infty))\right)>0$, then
\begin{equation*}\label{eq.het.NoPred.2}
p \geq \EE_{n}^{\alpha} \left( \lambda^a_{n+1}(\alpha) \,|\, D^-_{n+1}(p) > \nu^-_{n}((p,\infty)) \right);
\end{equation*}

\item for all $s>0$,

\begin{itemize}
\item $\lambda^a_n(\alpha) = \max\left\{ p^b_n,
\EE^{\alpha}_n \lambda^a_{n+1}(\alpha) + \left(\sup_{p\in\RR} \EE^{\alpha}_n \left( \left( p - \lambda^a_{n+1}(\alpha) \right) \bone_{\left\{ D^+_{n+1}(p) > \nu^+_{n}((-\infty,p)) \right\}} \right)\right)^+  \right\}$,

\item if $\hat{q}_n(s,\alpha)\neq 0$ and $\hat{r}_n(s,\alpha)=0$, then
$$
\lambda^a_n(\alpha) = 
\EE^{\alpha}_n \lambda^a_{n+1}(\alpha) + \sup_{p\in\RR} \EE^{\alpha}_n \left( \left( p - \lambda^a_{n+1}(\alpha) \right) \bone_{\left\{ D^+_{n+1}(p) > \nu^+_{n}((-\infty,p)) \right\}} \right),
$$
and $p=\hat{p}_n(s,\alpha)$ attains the above supremum,

\item if $\hat{q}_n(s,\alpha) = 0$ and $\hat{r}_n(s,\alpha)=0$, then $\lambda^a_n(\alpha) = \EE^{\alpha}_n \lambda^a_{n+1}(\alpha)$,

\item if $\hat{r}_n(s,\alpha)=1$, then $\lambda^a_n(\alpha) = p^b_n$;
\end{itemize}

\item for all $s<0$,

\begin{itemize}
\item $\lambda^b_n(\alpha) = \min\left\{ p^a_n,
\EE^{\alpha}_n \lambda^b_{n+1}(\alpha) - \left(\sup_{p\in\RR} \EE^{\alpha}_n \left( \left( \lambda^b_{n+1}(\alpha) - p \right) \bone_{\left\{ D^-_{n}(p) > \nu^-_{n-1}((p,\infty)) \right\}} \right)\right)^+  \right\}$,

\item if $\hat{q}_n(s,\alpha)\neq 0$ and $\hat{r}_n(s,\alpha)=0$, then
$$
\lambda^b_n(\alpha) = 
\EE^{\alpha}_n \lambda^b_{n+1}(\alpha) - \sup_{p\in\RR} \EE^{\alpha}_n \left( \left( \lambda^b_{n+1}(\alpha) - p \right) \bone_{\left\{ D^-_{n}(p) > \nu^-_{n-1}((p,\infty)) \right\}} \right),
$$
and $p=\hat{p}_n(s,\alpha)$ attains the above supremum,

\item if $\hat{q}_n(s,\alpha) = 0$ and $\hat{r}_n(s,\alpha)=0$, then $\lambda^b_n(\alpha) = \EE^{\alpha}_n \lambda^b_{n+1}(\alpha)$,

\item if $\hat{r}_n(s,\alpha)=1$, then $\lambda^b_n(\alpha) = p^a_n$.
\end{itemize}

\end{enumerate}
\end{cor}

\begin{proof}

Let us plug the piecewise-linear form of the value function into (\ref{eq.het.Vn})
\begin{equation*}
V^{\nu}_n(s,\alpha) = \text{esssup}_{p,q,r}\left\{\bone_{\left\{r_n=0\right\}} 
\left( s^+ \EE_n^{\alpha} \lambda^a_{n+1}(\alpha) - s^- \EE_n^{\alpha} \lambda^b_{n+1}(\alpha) 
\right.\right.
\end{equation*}
$$
\left.\left.
+ \EE_n^{\alpha} \left( \left(q_n p_n + (s-q_n)^+ \lambda^a_{n+1}(\alpha) - (s-q_n)^- \lambda^b_{n+1}(\alpha) - s^+ \lambda^a_{n+1}(\alpha) + s^- \lambda^b_{n+1}(\alpha) \right)\cdot
\right.\right.\right.
$$
$$
\left.\left.\left.
\left( \bone_{\left\{q_n\geq0,\,D^+_{n+1}(p_n) > \nu^+_n((-\infty,p_n)) \right\}} + \bone_{\left\{q_n<0,\,D^-_{n+1}(p_n) > \nu^-_n((p_n,\infty)) \right\}} \right)\right) \right)\right.
$$ 
$$
\left.
+ \bone_{\left\{r=1\right\}} \left( q^+_n p^b_n - q^-_n p^a_n + (s-q_n)^+ \EE_n^{\alpha} \lambda^a_{n+1}(\alpha) - (s-q_n)^- \EE_n^{\alpha} \lambda^b_{n+1}(\alpha) \right)
\right\}
$$
First, notice that it suffices to consider the essential supremum over all random variables $(p_n,q_n,r_n)$.\footnote{The admissibility constraint does not cause any difficulties here, as, in the case where $(p_n,q_n,r_n)$ do not attain the supremum, they can be improved, so that $(p_n,q_n)$ increase by no more than a fixed constant.}
Moreover, the essential supremum can be replaced by the supremum over all deterministic $(p_n,q_n,r_n)\in\RR^2\times\{0,1\}$. To see the latter, it suffices to assume that the supremum is not attained by the optimal strategy (with positive probability), and construct a superior strategy via the standard measurable selection argument (cf. Corollary 18.27 and Theorem 18.26 in \cite{Aliprantis}), which results in a contradiction.
It is easy to see that, for any fixed $(p_n,s,r_n)$, the above function is piece-wise linear in $q_n$, with the slope changing at $q_n=0$ and $q_n=s$. Hence, for a finite maximum to exists, the slope of this function must be nonnegative, at $q_n\rightarrow-\infty$, and non-positive, at $q_n\rightarrow \infty$. This must hold for any $(p_n,r_n,s)$, to ensure that the value function of an agent is finite: otherwise, an agent can scale up her position to increase the value function arbitrarily. Considering $r_n=1$, we obtain condition 2 of the corollary. The case $r_n=0$ yields conditions 3 and 4. Notice also that the maximum of the aforementioned function is always attained at $q_n=0$ or $q_n=s$. Considering all possible cases: $r_n=0,1$, $q_n=0,s$, $s=0$, $s>0$ and $s<0$ -- we obtain the recursive formulas for $\lambda^a_n$ and $\lambda^b_n$ (i.e., conditions  5 and 6 of the corollary). In addition, as the optimal $q_n$ takes values $0$ and $s$, it is easy to see that the piece-wise linear structure of the value function in $s$ is propagated backwards, and, hence, condition 1 of the corollary holds.
\qed
\end{proof}

It is also useful to have a converse statement.

\begin{cor}\label{cor:piecewiseLin.verif}
Consider an admissible LOB $\nu$ and admissible control $(\hat{p},\hat{q},\hat{r})$, such that $\hat{q}_n(s,\alpha)\in\left\{ 0,s\right\}$. Assume that, for any $\alpha\in\mathbb{A}$ and any $n=0,\ldots,N$, there exists a 
%$\mathcal{F}_n\otimes\mathcal{B}(\mathbb{\RR})$-measurable 
progressively measurable random function $V^{\nu}_{\cdot}(\cdot,\alpha)$, such that, for any $s\in\mathbb{\RR}$, $\PP$-a.s., $(\hat{p},\hat{q},\hat{r},V^{\nu})$ satisfy the conditions 1--6 of Corollary \ref{cor:piecewiseLin}. Then, $(\hat{p},\hat{q},\hat{r})$ is an optimal control for the LOB $\nu$.
\end{cor}

\begin{proof}
It suffices to revert the arguments in the proof of Corollary \ref{cor:piecewiseLin}, and recall that $\hat{q}$ can always be chosen to be equal to $0$ or $s$, without compromising the optimality.
\qed
\end{proof}

%The results of this section allow us to construct the continuum-player equilibria in specific models.

\section{Appendix B}

\emph{Proof of Lemma \ref{gapproxapplied}}.
The following lemma shows that the normalized price increments are close to Gaussian in the conditional $\mathbb{L}^2$ norm.

\begin{lemma}\label{l2conv}
Let Assumptions \ref{ass:sigma}, \ref{ass:A.alpha}, \ref{ass:joint.cond.reg}, \ref{ass:main.L2.strong} hold.
Then, there exists a deterministic function $\epsilon(\cdot)\ge0$, such that $\epsilon(\Delta t)\to0$, as $\Delta t\to0$, and, $\PP$-a.s., for all $\alpha\in\mathbb{A}$ and all $n=1,\ldots,N$, we have
$$
\EE^\alpha_{n-1}\left(\left(\xi_n/\sqrt{\Delta t} - \sigma_{t_{n-1}}(W^\alpha_{t_n}-W^\alpha_{t_{n-1}})/\sqrt{\Delta t}\right)^2\right) \le \epsilon(\Delta t).
$$
\end{lemma}
\begin{proof}
Notice: $\xi_n/\sqrt{\Delta t} - \sigma_{t_{n-1}}(W^\alpha_{t_n}-W^\alpha_{t_{n-1}})/\sqrt{\Delta t}
=\frac{1}{\sqrt{\Delta t}} \int\limits_{t_{n-1}}^{t_n}\mu^\alpha_s \text{d}s 
+ \frac{1}{\sqrt{\Delta t}} \int\limits_{t_{n-1}}^{t_n}(\sigma_s-\sigma_{t_{n-1}}) \text{d}W^\alpha_s$.
%and denote
%$$
%a_n=\frac{1}{\sqrt{\Delta t}} \int\limits_{t_{n-1}}^{t_n}\mu^\alpha_s \text{d}s,\quad b_n=\frac{1}{\sqrt{\Delta t}} \int\limits_{t_{n-1}}^{t_n}(\sigma_s-\sigma_{t_{n-1}}) \text{d}W^\alpha_s.
%$$ 
Then, using Assumptions \ref{ass:A.alpha}, \ref{ass:main.L2.strong}, and It\^{o}'s isometry, we obtain the statement of the lemma.
%As $|\mu^\alpha_s|\le C$, we see that $|a_n|\le C \sqrt{\Delta t}$ and, hence,
%$$
%\EE^{\alpha}_{t_{n-1}}\left[(a_n+b_n)^2\right]
%\le C^2\Delta t+2C\sqrt{\Delta t}\sqrt{\EE^{\alpha}_{t_{n-1}}\left[b_n^2\right]}+\EE^{\alpha}_{t_{n-1}}\left[b_n^2\right]
%$$
%To get the desired result it suffices to show that the conditional squared $\mathbb{L}^2$ norm $\EE^{\alpha}_{t_{n-1}}\left[b_n^2\right]\le \varepsilon(\Delta t)$, which follows easily from the It\^{o}'s isometry:
%$$
%\EE^{\alpha}_{t_{n-1}}\left[b_n^2\right] 
%= \frac{1}{\Delta t}\int\limits_{t_{n-1}}^{t_n}\EE^{\alpha}_{t_{n-1}}[(\sigma_s-\sigma_{t_{n-1}})^2]\text{d}s\le\varepsilon(\Delta t) ,
%$$
%using Assumption \ref{ass:main.L2.strong}.
\qed
\end{proof}

The next lemma connects the proximity in terms of $\mathbb{L}^2$ norm and the proximity of expectations of certain functions of random variables. This result would follow trivially from the classical theory, but, in the present case, we require additional uniformity -- hence, a separate lemma is needed (whose proof is, nevertheless, quite simple).

\begin{lemma}\label{gapprox}
For any constant $C>1$, there exists a deterministic function $\gamma(\cdot)\ge0$, s.t. $\gamma(\varepsilon)\to0$, as $\varepsilon\to0$, and, for any $\varepsilon>0$, $\sigma\in[1/C,C]$, and any random variables $\eta\sim\mathcal{N}(0,\sigma^2)$ and $\xi$ (the latter is not necessarily Gaussian), satisfying $\EE(\xi-\eta)^2\le\varepsilon$, the following holds for all $p\in\RR$
\begin{itemize}
\item[(i)] $(|p|\vee 1)\left| \PP(\xi>p) - \PP(\eta>p) \right|\le\gamma(\varepsilon)$,
\item[(ii)] $\left| \EE(\xi\bone_{\{\xi>p\}}) - \EE(\eta\bone_{\{\eta>p\}}) \right|\le\gamma(\varepsilon)$.
%\item[(iii)] $\left|\EE[p\bone_{\{\xi>p\}}]-\EE[p\bone_{\{\eta>p\}}]\right|\le\gamma(\varepsilon)$
\end{itemize}
\end{lemma}

\begin{proof}
\noindent (ii) Note that
$$
\left| \EE(\xi\bone_{\{\xi>p\}}) - \EE(\eta\bone_{\{\eta>p\}}) \right| \le
\left| \EE\left((\xi-\eta)\bone_{\{\xi>p\}}\right) \right| + \left|\EE\left(\eta(\bone_{\{\xi>p\}}-\bone_{\{\eta>p\}})\right)\right|
$$
$$
\leq \sqrt{\varepsilon} + \left\Vert\eta\right\Vert_2 \sqrt{\PP(\xi>p,\eta\le p) + \PP(\xi\le p,\eta>p)},
$$
and
$$
\PP(\xi>p,\eta\le p) \le \PP(p\ge\eta\ge p-\sqrt[3]{\varepsilon}) + \PP(|\xi-\eta|>\sqrt[3]{\varepsilon}) \le M\sqrt[3]{\varepsilon}+\frac{\EE(\xi-\eta)^2}{(\sqrt[3]{\varepsilon})^2}\le (M+1)\sqrt[3]{\varepsilon},
$$
where we used the fact that $\eta$ has a density bounded by a fixed constant $M$. We can similarly show that $\PP[\xi\le p,\eta>p]\le(M+1)\sqrt[3]{\varepsilon}$. The resulting estimates yield the statement of the lemma.
\qed
\end{proof}

Taking $\varepsilon(\Delta t)=\gamma(\epsilon(\Delta t))$ and applying the above lemmas, we get the statement of Lemma \ref{gapproxapplied}, with $(W^\alpha_{t_n}-W^\alpha_{t_{n-1}})/\sqrt{\Delta t}$ in place of $\eta_0$. Finally, we note that the laws of the two random variables coincide under $\PP^{\alpha}_{n-1}$, and the statement depends only on these laws. The last statement of Lemma \ref{gapproxapplied} follows from the fact that Lemma \ref{gapprox} is stable under analogous substitution.
\qed

\bibliography{MFGLOB_refs}

%\begin{figure}
%\begin{center}
%  \begin{tabular} {cc}
%    {
%    \includegraphics[width = 0.49\textwidth]{prices.png}
%    } & {
%    \includegraphics[width = 0.49\textwidth]{lambdas.png}
%    }\\
%  \end{tabular}
%  \caption{On the left: bid (negative) and ask (positive) prices. On the right: expected execution prices (negative values correspond to the agents trying to sell). Different curves correspond to different trading frequencies ($N=20,\ldots,500$). All prices are measured relative to the fundamental price and are plotted as functions of time. Zero drift case: $\alpha=0$, $\sigma=1$, $T=1$. }
%    \label{fig:1}
%\end{center}
%\end{figure}

\begin{figure}
\begin{center}
  \begin{tabular} {cc}
    {
    \includegraphics[width = 0.49\textwidth]{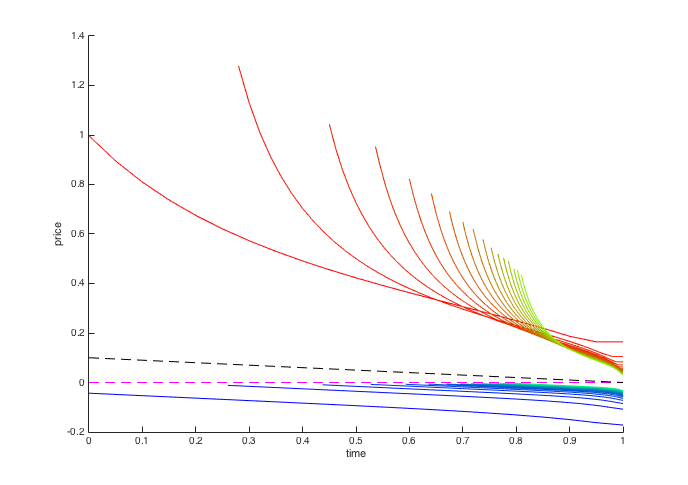}
    } & {
    \includegraphics[width = 0.49\textwidth]{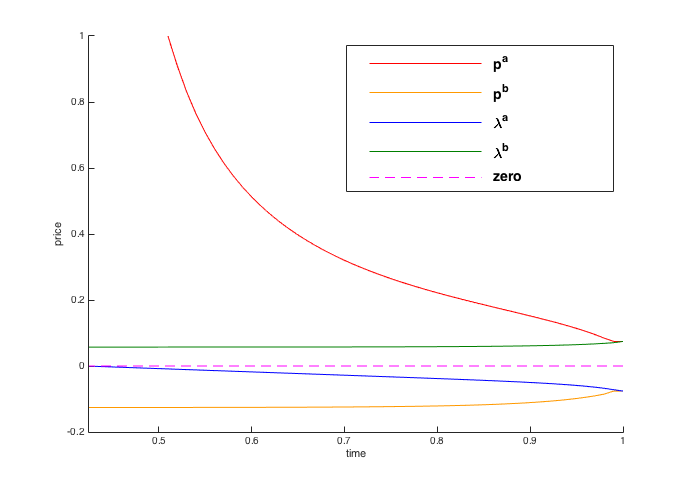}
    }\\
  \end{tabular}
  \caption{On the left: ask price $\hat{p}^a$ (in red) and the associated expected execution prices $\hat{\lambda}^a$ (in blue); different curves correspond to different trading frequencies ($N=20,\ldots,500$); black dashed line is the expected change in the fundamental price $\alpha (T-t)$. On the right: ask price $\hat{p}^a$ (in red) and the associated expected execution price $\hat{\lambda}^a$ (in blue), bid price $\hat{p}^b$ (in orange) and the associated expected execution price $\hat{\lambda}^b$ (in green), for $N=100$. Non-degenerate equilibrium exists only on a time interval where $\hat{\lambda}^a < 0$. All prices are measured relative to the fundamental price and are plotted as functions of time. Positive drift: $\alpha=0.1$, $\sigma=1$, $T=1$.}
    \label{fig:2}
\end{center}
\end{figure}

%\begin{figure}[!htb]
%\begin{center}
%\includegraphics[width = 0.7\textwidth]{explosion_2_alpha01.eps}
%\caption{Ask prices (in red) and the associated $\lambda^a$ (in blue), as functions of time. Different curves correspond to different trading frequencies. Positive drift case.} 
%\label{fig:2}
%\end{center}
%\end{figure}

\begin{figure}
\begin{center}
  \begin{tabular} {cc}
    {
    \includegraphics[width = 0.49\textwidth]{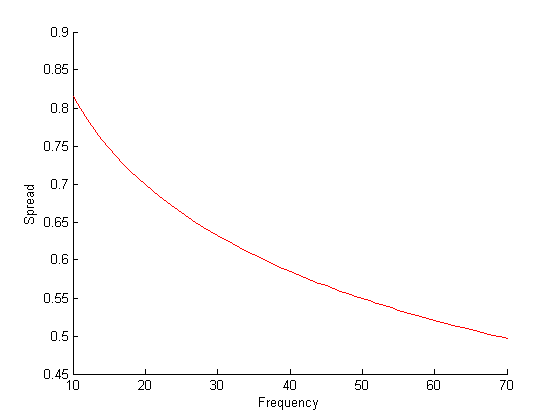}
    } & {
    \includegraphics[width = 0.49\textwidth]{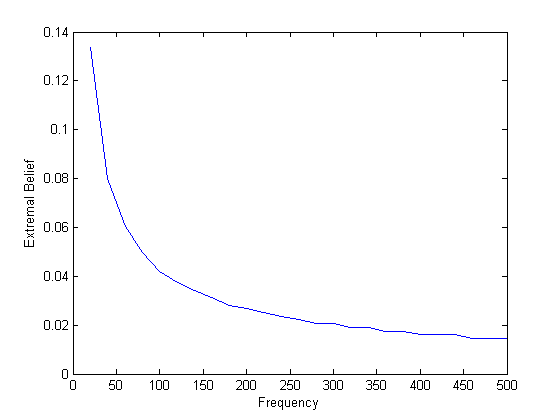}
    }\\
  \end{tabular}
  \caption{The horizontal axis represents trading frequency, measured in the number of steps $N$. Left: time-zero bid-ask spread in the zero-drift case ($\alpha=0$). Right: the maximum value of drift $\alpha$ for which a non-degenerate equilibrium exists on the entire time interval. Parameters: $\sigma=1$, $T=1$.}
    \label{fig:3}
\end{center}
\end{figure}

\end{document}